\documentclass[11pt,a4paper]{article}
\pdfoutput=1

\usepackage[vmargin=30mm,hmargin=33mm]{geometry}
\usepackage{amsmath,amssymb,amsthm,booktabs,color,doi,enumitem,framed,graphicx,latexsym,url,xcolor}
\usepackage[numbers,sort&compress]{natbib}
\usepackage{microtype,hyperref}
\usepackage{float}
\usepackage{booktabs}

\hypersetup{
    colorlinks=true,
    linkcolor=black,
    citecolor=black,
    filecolor=black,
    urlcolor=black,
}

\newtheorem{theorem}{Theorem}
\newtheorem{lemma}{Lemma}
\newtheorem{corollary}{Corollary} 
\theoremstyle{remark}
\newtheorem*{remark}{Remark} 

\newcommand{\namedref}[2]{\hyperref[#2]{#1~\ref*{#2}}}
\newcommand{\sectionref}[1]{\namedref{Section}{#1}}
\newcommand{\appendixref}[1]{\namedref{Appendix}{#1}}

\newcommand{\figureref}[1]{\namedref{Figure}{#1}}
\newcommand{\lemmaref}[1]{\namedref{Lemma}{#1}}
\newcommand{\tableref}[1]{\namedref{Table}{#1}}
\newcommand{\bftab}{\fontseries{b}\selectfont}

\newcommand{\cA}{\mathcal{A}}

\newcommand{\byz}{*}
\renewcommand{\vec}[1]{\mathbf{#1}}

\newcommand{\algorithm}[1]{%
\begin{framed}%
\begin{enumerate}%
#1
\end{enumerate}%
\end{framed}%
}

\newcommand{\block}[1]{%
\begin{itemize}[noitemsep,label=$\cdot$]%
#1
\end{itemize}%
}

\newenvironment{mycover}
               {\list{}{\listparindent 0pt
                        \itemindent    \listparindent
                        \leftmargin    0pt
                        \rightmargin   0pt
                        \parsep        0pt}%
                \raggedright
                \item\relax}
               {\endlist}

\begin{document}

\hypersetup{
    pdfauthor={Danny Dolev, Keijo Heljanko, Matti J\"arvisalo, Janne H. Korhonen, Christoph Lenzen, Joel Rybicki, Jukka Suomela, Siert Wieringa},
    pdftitle={Synchronous Counting and Computational Algorithm Design},
}

\begin{mycover}
{\LARGE \textbf{Synchronous Counting and\\Computational Algorithm Design}\par}

\bigskip
\bigskip
\textbf{Danny Dolev}\\
{\small School of Engineering and Computer Science,\\
The Hebrew University of Jerusalem\par}

\medskip
\textbf{Keijo Heljanko}\\
{\small Helsinki Institute for Information Technology HIIT, \\
Department of Computer Science and Engineering, Aalto University 
\par}

\medskip
\textbf{Matti J\"arvisalo}\\
{\small Helsinki Institute for Information Technology HIIT, \\
Department of Computer Science, University of Helsinki
\par}

\medskip
\textbf{Janne H.\ Korhonen}\\
{\small Helsinki Institute for Information Technology HIIT, \\
Department of Computer Science, University of Helsinki\par}

\medskip
\textbf{Christoph Lenzen}\\
{\small Department of Algorithms and Complexity, MPI Saarbr\"{u}cken\par}

\medskip
\textbf{Joel Rybicki}\\
{\small Helsinki Institute for Information Technology HIIT, \\
Department of Information and Computer Science, Aalto University\par}

\medskip
\textbf{Jukka Suomela}\\
{\small Helsinki Institute for Information Technology HIIT, \\
Department of Information and Computer Science, Aalto University\par}

\medskip
\textbf{Siert Wieringa}\\
{\small Helsinki Institute for Information Technology HIIT, \\
Department of Computer Science and Engineering, Aalto University \par}

\end{mycover}

\paragraph{Abstract.}

Consider a complete communication network on $n$ nodes, each of which is a state machine. In \emph{synchronous $2$-counting}, the nodes receive a common clock pulse and they have to agree on which pulses are ``odd'' and which are ``even''. We require that the solution is \emph{self-stabilising} (reaching the correct operation from any initial state) and it tolerates $f$ \emph{Byzantine failures} (nodes that send arbitrary misinformation). Prior algorithms are expensive to implement in hardware: they require a source of random bits or a large number of states.

This work consists of two parts. In the first part, we use computational techniques (often known as \emph{synthesis}) to construct very compact deterministic algorithms for the first non-trivial case of $f=1$. While no algorithm exists for $n < 4$, we show that as few as $3$ states per node are sufficient for all values $n \ge 4$. Moreover, the problem cannot be solved with only $2$ states per node for $n = 4$, but there is a $2$-state solution for all values $n \ge 6$.

In the second part, we develop and compare two different approaches for synthesising synchronous counting algorithms. Both approaches are based on casting the synthesis problem as a propositional satisfiability (SAT) problem and employing modern SAT-solvers. The difference lies in how to solve the SAT problem: either in a direct fashion, or incrementally within a \emph{counter-example guided abstraction refinement} loop. Empirical results suggest that the former technique is more efficient if we want to synthesise time-optimal algorithms, while the latter technique discovers non-optimal algorithms more quickly.

\newpage

\section{Introduction}\label{sec:intro}

\paragraph{Synchronous Counting.}

In the \emph{synchronous $C$-counting} problem, $n$ nodes have to count clock pulses modulo~$C$. Starting from any initial configuration, the system has to \emph{stabilise} so that all nodes agree on the clock value.
\begin{center}
    \includegraphics[page=1]{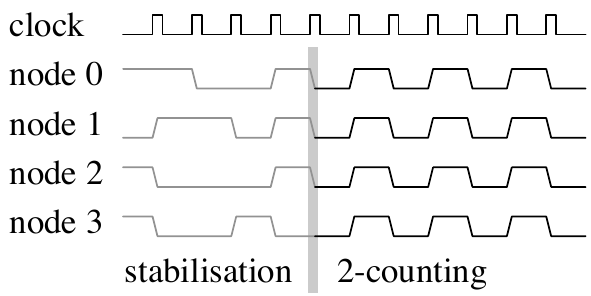}
\end{center}

Each node is a finite state machine with $s$ states, and after every state transition, each node \emph{broadcasts} its current state to all other nodes---effectively, each node can see the current states of all other nodes. An algorithm specifies (1)~the new state for each observed state, and (2)~how to map the internal state of a node to its output.

\paragraph{Byzantine Fault Tolerance.}

In a fault-free system, the $C$-counting problem is trivial to solve. For example, we can designate node $0$ as a leader, and then all nodes (including the leader itself) can follow the leader: if the current state of the leader is $c$, the new state is $c+1 \bmod C$. This algorithm will stabilise in time $t = 1$, and we only need $s = C$ different states.

However, we are interested in algorithms that tolerate \emph{Byzantine failures}. Some number $f$ of the nodes may be \emph{faulty}. A faulty node may send arbitrary misinformation to non-faulty nodes, including \emph{different} information to different nodes within the same round. For example, if we have nodes $0,1,2,3$ and node $2$ is faulty, node $0$ might observe the state vector $(0,1,1,1)$, while node $1$ might observe the state vector $(0,1,0,1)$.

Our goal is to design an algorithm with the following guarantee: even if we have up to $f$ faulty nodes, no matter what the faulty nodes do, the system will stabilise so that after $t$ rounds all non-faulty nodes start to count clock pulses consistently modulo $C$. We will give a formal problem definition in \sectionref{sec:model}.

\begin{center}
    \includegraphics[page=2]{figs.pdf}
\end{center}

\paragraph{State of the Art.}

Both randomised and deterministic algorithms for synchronous counting have been presented in the literature (see \sectionref{sec:related}). However, prior algorithms tend to be expensive to implement in hardware: they require a source of random bits or complicated circuitry.

In this work, we use a single parameter $s$, the number of states per node, to capture the complexity of an algorithm. If one resorts to randomness, it is possible to solve $2$-counting with the trivially optimal number of $s=2$ states---at the cost of a slow stabilisation time (see Sections \ref{sec:related} and~\ref{sec:human-algs}). However, it is not at all clear whether a small number of states suffices for \emph{deterministic} algorithms.

\paragraph{Contributions.}

We employ \emph{computational} techniques to design deterministic $2$-counting algorithms that have the smallest possible number of states. Our contributions are two-fold:
\begin{enumerate}
    \item we present new algorithms for the synchronous counting problem,
    \item we develop new computational techniques for constructing self-stabilising Byzantine fault-tolerant algorithms.
\end{enumerate}

Our focus is on the first non-trivial case of $f = 1$. The case of $n = 1$ is trivial, and by prior work it is known that there is no algorithm for $1 < n < 4$. We give a detailed analysis of $2$-counting for $n \ge 4$:
\begin{itemize}[noitemsep]
    \item there is no deterministic algorithm for $f = 1$ and $n = 4$ with $s = 2$ states,
    \item there is a deterministic algorithm for $f = 1$ and $n \ge 4$ with $s = 3$ states,
    \item there is a deterministic algorithm for $f = 1$ and $n \ge 6$ with $s = 2$ states.
\end{itemize}
Overall, we develop more than a dozen different algorithms with different characteristics, each of which can be also generalised to a larger number of nodes. See \figureref{fig:states} for an overview of the time--space tradeoffs that we achieve with our algorithms.

\begin{figure}
    \centering
    \includegraphics{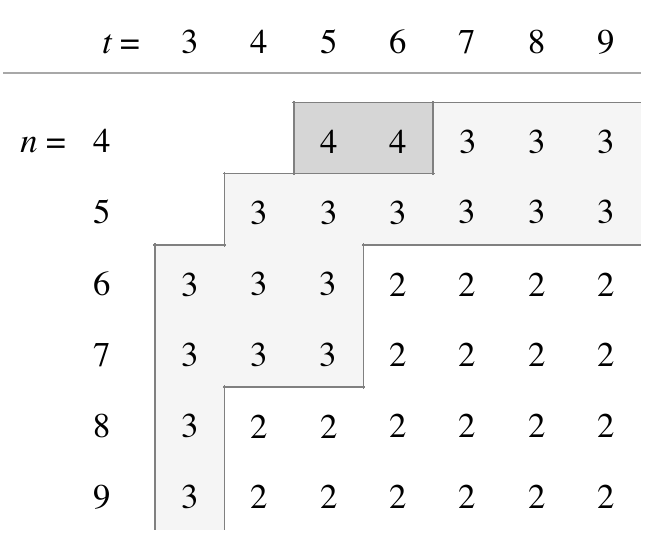}
    \caption{Time--space tradeoffs in our computer-designed algorithms. The figures shows $s$ (the number of states) for each combination of $n$ (the number of node) and $t$ (the stabilisation time).}\label{fig:states}
\end{figure}

With very few states per node, our algorithms are easy to implement in hardware. For example, a straightforward implementation of our algorithm for $f = 1$, $n = 4$, and $s = 3$ requires just $2$ bits of storage per node, and a lookup table with $81$ entries. All of our computer-designed algorithms are freely available online~\cite{results} in a machine-readable format. While our algorithms are synchronous $2$-counters, they can be easily composed to construct synchronous $2^b$-counters (see \sectionref{sec:app} for details).

This work can be seen as a case study of applying synthesis techniques in the area of distributed algorithms. We demonstrate that the synthesis of non-trivial self-stabilising Byzantine fault-tolerant algorithms is indeed possible with the help of modern propositional satisfiability (SAT) solvers~\cite{een04minisat,biere14lingeling}. We describe two complementary approaches for the synthesis of synchronous $2$-counting algorithms and give an empirical comparison of their relative performance:
\begin{enumerate}[noitemsep]
    \item a direct encoding as SAT,
    \item a SAT-based counter-example guided abstraction refinement (CEGAR)~\cite{clarke03cegar,clarke04sat-cegar} approach.
\end{enumerate}
Both approaches make it possible to use modern SAT solvers and to benefit from the steady progress in SAT solver technology. As we will see, the former approach is typically more efficient for tightly-specified problems (e.g., synthesising both space-optimal and time-optimal algorithms), while the latter is more promising for more relaxed problems (e.g., synthesising space-optimal algorithm regardless of the stabilisation time).

\paragraph{Structure.}

\sectionref{sec:related} covers related work and \sectionref{sec:app} discusses applications of synchronous $2$-counters. \sectionref{sec:model} gives a formal definition of the problem, and \sectionref{sec:human-algs} gives two examples of human-designed algorithms. \sectionref{sec:projgraph} gives a graph-theoretic interpretation that is helpful in the analysis of counting algorithms. In \sectionref{sec:smalln} we show that (1)~we can increase $n$ for free, without affecting the parameters $f$, $s$, or $t$; this enables us to focus on small values of~$n$, and (2)~we can generalise the algorithms to a larger class of network topologies with a slight cost in stabilisation time. \sectionref{sec:machine-pos} presents an overview of the use of computers in algorithm design and highlights the new results for synchronous counting. \sectionref{sec:sat-synth} describes a direct formulation of the synthesis problem for synchronous counting algorithms as propositional satisfiability. \sectionref{sec:ceg-synth} describes the SAT-based counter-example guided abstraction refinement synthesis technique. Finally, \sectionref{sec:comparison-results} overviews the results of the empirical evaluation of the two different synthesis techniques, suggesting a tradeoff between establishing the existence of any algorithm and finding optimal algorithms.

\section{Related Work}\label{sec:related}

\paragraph{Randomised Algorithms for Synchronous Counting.}

Randomised algorithms for synchronous $2$-counting are known, with different time--space tradeoffs.

The algorithm by Dolev and Welch~\cite{dolev04clock-synchronization} requires only $s=3$ states, but the stabilisation time is $t=2^{O(f)}$. Here we are assuming that $n = O(f)$; for a large $n$, we can run the algorithm with $O(f)$ nodes only and let the remaining nodes follow the majority.

The algorithm by Ben-Or et al.~\cite{ben-or08fast} stabilises in expected constant time. However, it requires $\Omega(2^f)$ states and private channels (i.e., the adversary has limited information on the system's state).

\paragraph{Deterministic Algorithms for Synchronous Counting.}

The fastest known deterministic algorithm is due to Dolev and Hoch~\cite{dolev07actions}, with a stabilisation time of $O(f)$. However, the algorithm is not well suited for a hardware implementation. It uses as a building block several instances of algorithms that solve the Byzantine consensus problem---a non-trivial task in itself. The number of states is also large, as some storage is needed for each Byzantine consensus instance.

\paragraph{Consensus Lower Bounds for Synchronous Counting.}

\emph{Binary consensus} is a classical problem that has been studied in the context of Byzantine fault tolerance; see, e.g., the textbook by Lynch~\cite{lynch96book} for more information. In brief, the problem is defined as follows. Each node has a binary input, and all non-faulty nodes have to produce the same binary output, $0$ or $1$. If all inputs are equal to $0$, the common output has to be $0$, and if all inputs are equal to $1$, the common output has to be $1$; otherwise the common output can be either $0$ or $1$. It is easy to show that synchronous $2$-counting is at least as difficult to solve as binary consensus.

\begin{lemma}
    If we have a $2$-counting algorithm $\cA$ that stabilises in time $t$, we can design an algorithm that solves binary consensus in time $t$, for the same parameters $n$ and $f$.\label{lemma:reduction}
\end{lemma}
\begin{proof}
Let $\vec{x}(0)$ and $\vec{x}(1)$ be some configurations that may occur during the correct operation of $\cA$ after it has stabilised, so that in configuration $\vec{x}(a)$ all nodes output $a$. More specifically:
\begin{itemize}
    \item For any $a = 0, 1$ and $j = 0,1,2,\dotsc$, if we initialise the system with configuration $\vec{x}(a)$ and run $\cA$ for $j$ rounds, all non-faulty nodes output $(a+j) \bmod 2$.
\end{itemize}
First assume that $t$ is even. Each node $i$ receives its input $a$ for the binary consensus problem. We use the element $i$ of $\vec{x}(a)$ to initialise the state of node $i$. Then we run $\cA$ for $t$ rounds. Finally, the output of algorithm $\cA$ forms the output of the binary consensus instance. To see that the algorithm is correct, we make the following observations:
(1)~All non-faulty nodes produce the same output at time $t$, regardless of the input.
(2)~If all inputs had the same value $a$, we used $\vec{x}(a)$ to initialise all nodes, and hence the final output is $a$.

For an odd $t$, we can use the same approach if we complement the inputs. In summary, $\cA$ can be used to solve binary consensus in time $t$.
\end{proof}

\noindent Now we can invoke the familiar lower bounds related to the consensus problem:
\begin{itemize}[noitemsep]
    \item no algorithm can tolerate $f \ge n/3$ failures \cite{pease80reaching},
    \item no deterministic algorithm can solve the problem in $t < f+1$ rounds \cite{fischer82lower}.
\end{itemize}

\paragraph{Pulse Synchronisation.}

Both $2$-counting and \emph{pulse synchronisation}~\cite{daliot03self-stabilizing,dolev04clock-synchronization} have a superficially similar goal: produce well-separated, (approximately) synchronised clock pulses in a distributed system in a fault-tolerant manner. However, there are also many differences: in pulse synchronisation the task is to construct a clock pulse without any external reference, while in $2$-counting we are given a reference clock and we only need to construct a clock that ticks at a slower rate. Also the models of computation differ---for pulse synchronisation, a relevant model is an asynchronous network with some bounds on propagation delays and clock drifts.

A $2$-counting algorithm does not solve the pulse synchronisation problem, and a pulse synchronisation algorithm does not solve the $2$-counting problem. However, if one is designing a distributed system that needs to produce synchronised clock ticks in a fault-tolerant manner, either of the approaches may be applicable.

\paragraph{Computational Algorithm Design.}

The computational element of our work can be interpreted as a form of \emph{algorithm synthesis}. In synthesis, the task is to algorithmically find an algorithm or a protocol that satisfies a given specification. The idea of synthesising circuits was proposed by e.g. Church~\cite{church62logic} already in the 1960s and there exists a vast body of work related to synthesis.

Classic work on model checking~\cite{clarke82design,manna84synthesis} consider algorithms for synthesis of both shared-memory and message-passing protocols by solving the satisfiability of certain temporal logic formulas. Unfortunately, synthesis of distributed systems is often intractable both in theory and practice---distributed synthesis problems are often either of high complexity or undecidable~\cite{pnueli90distributed,naor95what,finkbeiner05uniform}. However, despite the hardness of synthesis---or because of it---several techniques have been proposed to make synthesis tractable~\cite{finkbeiner12lazy,jacobs12parameterized,finkbeiner13bounded}. 

In contrast to applying general synthesis techniques, that is, algorithms for synthesising a general class of problems, combinatorial search algorithms have also been applied to solve specific synthesis problems. For example, SAT solvers have been used for, e.g., circuit synthesis~\cite{kojevnikov09finding,grosse09synthesis,fuhs10synthesizing,jarvisalo12finding,bloem14sat}, synthesis from safety specifications~\cite{bloem14synthesis}, controller synthesis~\cite{morgenstern13games}, program sketching~\cite{solar06sketching}, synthesising sorting networks~\cite{morgenstern11synthesis,bundala14sorting,codish14comparators}, and synthesising local graph algorithms~\cite{rybicki11msc,hirvonen14local-maxcut}.

\section{Applications}\label{sec:app}

\paragraph{Counters as Frequency Dividers.}

We can visualise a $C$-counter as an electronic circuit that consists of $n$ components (nodes); see \figureref{fig:circuit}. Each node $i$ has a register $x_i$ that stores its current state---one of the values $0,1,\dotsc,s-1$. There is a logical circuit $g$ that maps the current state to the output, and another logical circuit $A_i$ that maps the current states of all nodes to the new state of node $i$. At each rising edge of the clock pulse, register $x_i$ is updated.

\begin{figure}[h]
    \centering
    \includegraphics[page=5]{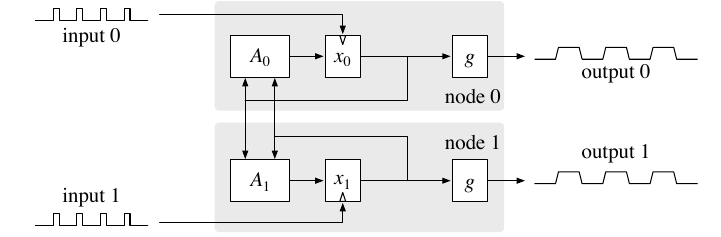}
    \caption{A $2$-counter for $n = 2$, viewed as an electronic circuit.}\label{fig:circuit}
\end{figure}

If the clock pulses are synchronised, regardless of the initial states of the registers, after $t$ clock pulses the system has stabilised so that the outputs are synchronised and they are incremented (modulo $C$) at each clock pulse.

In particular, if we have an algorithm for $2$-counting, it can be used as a \emph{frequency divider}: given synchronous clock pulses at rate $1$, it produces synchronous clock pulses at rate $1/2$.

\paragraph{\boldmath From $2$-Counters to $C$-Counters.}

We can compose $b$ layers of $2$-counters to build a clock that counts modulo $2^b$; see \figureref{fig:compose}. A composition of self-stabilising algorithms is self-stabilising~\cite{dolev00self-stabilization}. For the purposes of the analysis, we can wait until layer $i-1$ stabilises, use this as the initial state of layer $i$, and then argue that the nodes on layer $i$ receive a synchronous clock pulse and hence they will eventually stabilise.

\begin{figure}[h]
    \centering
    \includegraphics[page=4]{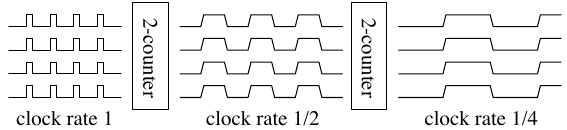}
    \caption{Composition of $2$-counters.}\label{fig:compose}
\end{figure}

\paragraph{Counters in Mutual Exclusion.}

With a $C$-counter we can implement \emph{mutual exclusion} and \emph{time division multiple access} in a fairly straightforward manner. If we have $C = n$ nodes and one shared resource (e.g., a transmission medium), we can let node $i$ to access the resource when its own counter has value $i$. Care is needed with the actions of faulty nodes, though---for further information on achieving \emph{fault-tolerant} mutual exclusion, see, e.g., Moscibroda and Oshman \cite{moscibroda11resilience}. Again $2$-counting is of particular interest, as it may be leveraged by more complex mutual exclusion algorithms.

\section{Problem Formulation}\label{sec:model}

We will now formalise the $C$-counting problem and the synthesis problem, and introduce the definitions that we will use in this work. Throughout this work, we will follow the convention that nodes, states, and time steps are indexed from $0$. We use the notation $[k] = \{0, 1, \dots, k-1 \}$.

\paragraph{Simplifications.}

As our focus is primarily on $2$-counters, we will now fix $C = 2$; the definitions are straightforward to generalise.

In prior work, algorithms have made use of a function that maps the internal state $x_i$ of a node to its output $g(x_i)$. However, in this work we synthesise algorithms that do not need any such mapping: for our positive results, an identity mapping is sufficient, and for the negative result, we study the case of $s=2$ which never benefits from a mapping. Hence we will now give a formalisation that omits the output mapping.

\paragraph{Algorithms.}

Fix the following parameters:
\begin{itemize}[noitemsep]
    \item $n$ = the number of nodes,
    \item $f$ = the maximum number of faulty nodes,
    \item $s$ = the number of internal states.
\end{itemize}
An algorithm $\vec A$ specifies a \emph{state transition} function $A_i \colon [s]^n \to [s]$ for each node $i \in [n]$. Here $[s]^n$ is the set of \emph{observed configurations} of the system.

\paragraph{Projections.}

Let $F \subseteq [n]$, $|F| \le f$ be the set of \emph{faulty} nodes. We define the \emph{projection} $\pi_F$ as follows: for any observed configuration $\vec{u} \in [s]^n$, let $\pi_F(\vec{u})$ be a vector $\vec{x}$ such that $x_i = \byz$ if $i \in F$ and $x_i = u_i$ otherwise. For example,
\[
    \pi_{\{2,4\}} ((0,1,0,1,1)) = (0,1,\byz,1,\byz).
\]
This gives us the set $V_F = \pi_F([s]^n)$ of \emph{actual configurations}. Two actual configurations are particularly important:
\[
    \vec 0_F = \pi_F((0,0,\dotsc,0)) \quad \text{and} \quad
    \vec 1_F = \pi_F((1,1,\dotsc,1)).
\]

\paragraph{Executions.}

Let $\vec x, \vec y \in V_F$. We say that configuration $\vec y$ is \emph{reachable} from $\vec x$ if for each non-faulty node $i \notin F$ there exists some observed configuration $\vec{u}_i \in [s]^n$ satisfying $\pi_F(\vec{u}_i) = \vec x$ and $A_i(\vec{u}_i) = y_i$. Intuitively, the faulty nodes can feed such misinformation to node $i$ that it chooses to switch to state $y_i$. We emphasise that $\vec{u}_i$ may be different for each $i$; the misinformation need not be consistent.

An \emph{execution} of an algorithm $\vec A$ for given set of faulty nodes $F$ is an infinite sequence of actual configurations $X = (\vec{x}^0, \vec{x}^1, \vec{x}^2, \dotsc)$ such that $\vec{x}^{r+1}$ is reachable from $\vec{x}^r$ for all $r$.

\paragraph{Stabilisation.}

For an execution $X = (\vec{x}^0, \vec{x}^1, \vec{x}^2, \dotsc)$, define its $t$-tail \[X[t] = (\vec{x}^t, \vec{x}^{t+1}, \vec{x}^{t+2}, \dotsc).\] We say that $X$ \emph{stabilises in time $t$} if one of the following holds:
\[
    X[t] = (\vec 0_F, \vec 1_F, \vec 0_F, \dotsc) \quad \text{or} \quad
    X[t] = (\vec 1_F, \vec 0_F, \vec 1_F, \dotsc).
\]
We say that an algorithm $\vec A$ \emph{stabilises in time $t$} if for any set of faulty nodes $F$ with $|F| \le f$, all executions of $\vec A$ stabilise in time~$t$.

\paragraph{The Synthesis Problem.}

Now that we have formally defined what a 2-counting algorithm is, we can give the definition for the synthesis problem of counting algorithms. First, the decision version of the problem is the \emph{realisability problem}. Given an instance $(n,f,s,t)$, the task is to decide whether there exists a 2-counting algorithm for a network with $n$ nodes satisfying the following properties:
\begin{enumerate}[noitemsep]
    \item the algorithm tolerates $f$ failures,
    \item each node uses at most $s$ states,
    \item the algorithm stabilises in at most $t$ steps.
\end{enumerate}
If such an algorithm exists, we say that the instance $(n, f, s, t)$ is \emph{realisable}. The \emph{synthesis problem} is to output an algorithm $\vec A$ if the instance is realisable or state that no algorithm exists.

\section{Human-Designed Algorithms}\label{sec:human-algs}

Before moving on to computer-designed algorithms using SAT-based techniques, in this section we illustrate a few human-designed algorithms. First, we show that randomisation helps when it comes to designing small-state (but slow) algorithms. This is followed by a deterministic algorithm that solves the counting problem in the general case with a large number of internal states.

\paragraph{Randomised Algorithms.}

We extend our model to randomised algorithms by equipping each node with a \emph{private coin}. Now in a single synchronous round, every node can flip its coin to access one random bit. Thus, node $i$ can decide on its new state using the random bit $b \in \{0,1\}$ and the observed configuration $\vec u \in [s]^n$. In contrast to the randomised algorithm by Dolev and Welch~\cite{dolev04clock-synchronization}, the following algorithm only uses two states.

Let $n \ge 4$, $f < n/3$, and $s=2$. We can solve the 2-counting problem with the algorithm of \figureref{fig:alg-rand}.

\begin{figure}
\algorithm{
    \item If more than $(n+f)/2$ entries in $\vec u$ are 0:
    \block{
        \item Switch to state $1$.
    }
    \item Otherwise, if more than $(n+f)/2$ entries in $\vec u$ are 1:
    \block{
        \item Switch to state $0$.
    }
    \item Otherwise:
    \block{
        \item Flip the coin to get a random bit $b \in \{0,1\}$.
        \item Switch to state $b$.
    }
    \vspace{-3mm}
}
\vspace{-4mm}
\caption{A randomised 2-counting algorithm. All nodes follow the same algorithm.}\label{fig:alg-rand}
\end{figure}

\begin{lemma}
 Let $p$ be the probability that out of $n - f - 1$ coin flips, more than $(n+f)/2 - 1$ are heads. 
 The randomised algorithm solves synchronous 2-counting in $1/p + 1$ rounds in expectation.\label{lemma:expected}
\end{lemma}
\begin{proof}
 Observe that no two distinct non-faulty nodes apply rules 1 and 2 during the same round: if a node $i$ sees the value 0 more than $(n+f)/2$ times, then any node $j$ must see value 0 at least $(n-f)/2$ times, and thus, $j$ sees the value 1 fewer than $(n+f)/2$ times. Moreover, if more than $(n+f)/2$ non-faulty nodes have the same output, then the system will stabilise in the next round as all non-faulty nodes switch to the same state.

 Next we argue that with probability at least $p$, more than $(n+f)/2$ non-faulty nodes have the same state. We have three cases. In the first case, at least one non-faulty node applies rule 1. Then in the worst case all other nodes flip their coins, so the system stabilises with probability at least $p$. The second case, where at least one non-faulty node applies rule 2, is symmetrical. Finally, the third case consists of all nodes flipping their coins simultaneously. In this case, fix the output of a single non-faulty node and repeat the analysis of the previous two cases.

 The number of rounds before we stabilise follows a geometric distribution, so in expectation, we get a successful streak of coin flips in $1/p$ rounds and stabilise during the next round. 
\end{proof}

\begin{theorem}
 For all $n \ge 4$ and $f \le n/3$, the expected stabilisation time of the randomised algorithm is bounded by
 $$
 \min \{ 2^{2f+2} + 1, 2^{O(f^2/n)}\}.
 $$
\end{theorem}
\begin{proof}
 We bound the probability $p$ in Lemma~\ref{lemma:expected} from which the expected stabilisation time follows.

 For the first bound, it suffices to analyse the event where the first $2f+1$ non-faulty nodes and at least half of the remaining non-faulty nodes all flip head at the same round, as $2f+1 - (n-f-2f-1)/2 > (n+f)/2$. Now observe that the probability of $2f+1$ coin flips all being head is $2^{-2f-1}$ and the probability that at least half of out of $N$ coin flips are head is at least $1/2$. Combining these observations gives us the first bound.

 For the second bound, if $f = \Theta(n)$ then the second bound trivially follows from the first. Suppose $f = o(n)$. We use the fact~\cite{feller43generalization,matousek08probabilistic-lecture} that for any $t \in [N/8]$ 
$$\Pr[X \ge N/2 + t] \ge \frac{1}{15}\exp(-16t^2/N),$$
where $X$ is the number of heads in $N$ coin flips. Setting $N=n-f-1$ and $t=\lfloor (n+f)/2 \rfloor + 1 - N/2$ gives us the desired bound.
\end{proof}

\paragraph{Deterministic Algorithms.}

We can leverage existing deterministic algorithms for binary consensus to come up with synchronous counting algorithms. However, this leads to a large number of states per node.

\pagebreak % LAYOUT

For example, this theorem follows from the results by Dolev and Hoch~\cite{dolev07actions}:
\begin{theorem}
 Let $\vec{A}$ be a deterministic algorithm that solves binary consensus in $R$ rounds for $n$ nodes and $f$ faults. Then there exists a deterministic algorithm $\vec{B}$ that solves synchronous $C$-counting in time $t \in O(R+C)$ for $n$ nodes and $f$ faults.
\end{theorem}

Now we can use any consensus algorithm, such as the phase king algorithm~\cite{berman89consensus}, to get a synchronous counter. The phase king achieves optimal resilience and has $O(f)$ stabilisation time and uses $O(\log f)$ state bits (for keeping track of the current round number) per node. However, the resulting synchronous counter relies on executing $O(f)$ consensus instances in parallel, which yields into a very large state space. We get the following corollary:
\begin{corollary}
 For all $n \ge 4$, $f < n/3$ and $C \ge 2$, there is a deterministic $C$-counting algorithm that stabilises in $t \in O(C+f)$ rounds and uses $s \in 2^{O(\log C + f \log f)}$ states.
\end{corollary}

This approach is not very attractive, for example, from the perspective of hardware implementations. We will now turn our attention to efficient, deterministic, computer-designed algorithms.

\section{Projection Graphs}\label{sec:projgraph}

Before discussing how to \emph{find} an algorithm (or prove that an algorithm does not exist), let us first explain how we can \emph{verify} that a given algorithm is correct. Here the concept of a \emph{projection graph} is helpful---see \figureref{fig:alg-3-4-1-7-c} in the appendix for an example.

Fix the parameters $s$, $n$, and $f$, and consider a candidate algorithm $\vec A$ that is supposed to solve the $2$-counting problem. For each set $F \subseteq [n]$ of faulty nodes, construct the directed graph $G_F(\vec A) = (V_F, R_F(\vec A))$ as follows.
\begin{enumerate}
    \item The set of nodes $V_F$ is the set of actual configurations.
    \item There is an edge $(\vec u, \vec v) \in R_F(\vec A)$ if configuration $\vec v \in V_F$ is reachable from configuration $\vec u \in V_F$. In general, this may produce self-loops.
\end{enumerate}
Note that the outdegree of each node in $G_F(\vec A)$ is at least $1$. Directed walks in $G_F(\vec A)$ correspond to possible executions of algorithm $\vec A$, for this set $F$ of faulty nodes. To verify the correctness of algorithm $\vec A$, it is sufficient to analyse the projection graphs $G_F$. The following lemmas are straightforward consequences of the definitions.

\begin{lemma}\label{lem:proj}
    Algorithm $\vec A$ stabilises in some time $t$ iff for every $F$, graph $G_F(\vec A)$ contains exactly one directed cycle, $\vec 0_F \mapsto \vec 1_F \mapsto \vec 0_F$.
\end{lemma}

\begin{lemma}\label{lem:projt}
    Algorithm $\vec A$ stabilises in time $t$ iff the following holds for all~$F$:
    \begin{enumerate}
        \item In $G_F(\vec A)$, the only successor of $\vec 0_F$ is $\vec 1_F$ and vice versa.
        \item In $G_F(\vec A)$, every directed walk of length $t$ reaches node $\vec 0_F$ or $\vec 1_F$.
    \end{enumerate}
\end{lemma}

\begin{lemma}\label{lem:mix}
    Let $\vec A$ be an algorithm. Consider any four configurations $\vec x, \vec u, \vec v, \vec w \in V_F$ with the following properties:
    $(\vec x, \vec u) \in R_F(\vec A)$,
    $(\vec x, \vec v) \in R_F(\vec A)$, and
    $w_i \in \{u_i, v_i\}$ for each $i \notin F$.
    Then $(\vec x, \vec w) \in R_F(\vec A)$.
\end{lemma}

\section{Increasing the Number of Nodes}\label{sec:smalln}

It is not obvious how to use computational techniques to design an algorithm that solves the $2$-counting problem for a fixed $f = 1$ but arbitrary $n \ge 4$. However, as we will show next, we can generalise any algorithm so that it solves the same problem for a larger number of nodes, without any penalty in time or space complexity. Therefore it is sufficient to design an algorithm for the special case of $f = 1$ and $n = 4$. From the perspective of parametrised verification and synthesis, the following lemma can be regarded as a \emph{cut-off} result~\cite{emerson95reasoning,jacobs12parameterized}.

\begin{lemma}
    Fix $n \ge 4$, $f < n/2$, $s \ge 2$, and $t \ge 1$. Assume that $\vec A$ is an algorithm that solves the $2$-counting problem for $n$ nodes, out of which at most $f$ are faulty, with stabilisation time $t$ and with $s$ states per node. Then we can design an algorithm $\vec B$ that solves the $2$-counting problem for $n+1$ nodes, out of which at most $f$ are faulty, with stabilisation time $t$ and with $s$ states per node.\label{lemma:increase}
\end{lemma}
\begin{proof}
The claim would be straightforward if we permitted the stabilisation time of $t+1$. However, some care is needed to avoid the loss of one round.

We take the following approach. Let $p$ be a projection that removes the last element from a vector, for example, $p((a,b,c)) = (a,b)$. In algorithm $\vec B$, nodes $i \in [n]$ simply follow algorithm $\vec A$, ignoring node $n$:
\[
    B_i(\vec{u}_i) = A_i(p(\vec{u}_i)).
\]
Node $n$ tries to predict the majority of nodes $0,1,\dotsc,n-1$, i.e., what most of them are going to output after this round:
\begin{itemize}
    \item Assume that node $n$ observes a configuration $\vec{u}_n$. For each $i \in [n]$, define
    $h_i = A_i(p(\vec{u}_n))$. If a majority of the values $h_i$ is $1$, then the new state of node $n$ is also $1$; otherwise it is $0$.
\end{itemize}

To prove that the algorithm is correct, fix a set $F \subseteq [n+1]$ of faulty nodes, with $|F| \le f$. Clearly, all nodes in $[n]\setminus F$ will start counting correctly at the latest in round~$t$. Hence any execution of $\vec B$ with $n \in F$ trivially stabilises within $t$ rounds; so we focus on the case of $F \subseteq [n]$, and merely need to show that also node $n$ counts correctly.

Fix an execution $X=(\vec{x}^0,\vec{x}^1,\ldots)$ of $\vec A$, and a point of time $r \ge t$. Consider the state vector $\vec{x}^{r-1}$. By assumption, $\vec A$ stabilises in time $t$. Hence the successors of $\vec{x}^{r-1}$ in the projection graph must be in $\{\vec 0_F, \vec 1_F\}$.

The key observation is that only one of the configurations $\vec 0_F$ and $\vec 1_F$ can be the successor of $\vec{x}^{r-1}$. Otherwise \lemmaref{lem:mix} would allow us to construct another state that is a successor of $\vec{x}^{r-1}$, contradicting the assumption that $\vec A$ stabilises.

We conclude that for all rounds $r\geq t$ and all nodes $i\in [n]\setminus F$, the value $h_i$ is independent of the states communicated by nodes in $F$. Since the values $h_i$ are identical and $n - f > f$, node $n$ attains the same state as other correct nodes in rounds $r\geq t$.
\end{proof}

\paragraph{Other Network Topologies.}

Next we show that it is relatively straightforward to generalise our small-state algorithms to other network topologies as well---albeit with a slight increase in the stabilisation time. The idea is to have a small core of nodes to initially solve synchronous counting, and from thereon, propagate the solution throughout the network. This approach was originally introduced by Braud-Santoni et al.~\cite{braud-santoni14synthesising}. We now show how this idea can be applied in a large class of graphs.

Consider the following families of graphs $\mathcal{G}(k,m,d)$ for integers $k,m,d>0$. Let $G = (V,E)$ be a graph. We say $G \in \mathcal{G}(k,m,d)$ if there exists a partition $V_0, \dots, V_d$ of the nodes $V$ such that 
\begin{enumerate}
 \item $V_0$ is a $k$-clique.
 \item Each node $i \in V_a$ has at least $m$ neighbours in $V_0, \dots, V_{a-1}$.
\end{enumerate}

Put otherwise, we can characterise $\mathcal{G}(k,m,d)$ using the following game (which is reminiscent of threshold models in the context of influence spreading in social networks). Initially, colour all vertices of graph $G$ white. We pick a clique of $k$ nodes and colour all the nodes black. Now any node with at least $m$ black neighbours switches its own colour black. If after $d$ iterations all nodes are coloured black, then $G \in \mathcal{G}({k,m,d})$. See \figureref{fig:topologies} for examples.

\begin{figure}
    \centering
    \includegraphics[page=6]{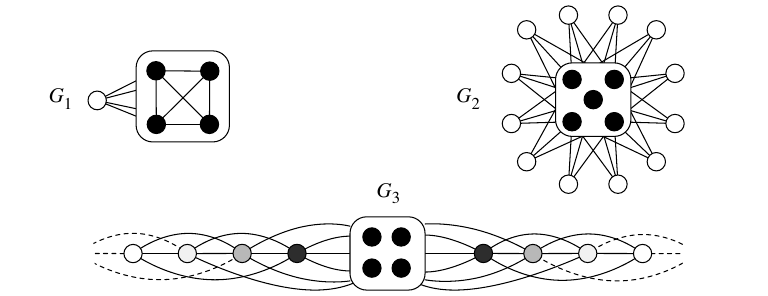}
    \caption{Examples of generalised network topologies. Nodes encompassed within a rectangle form a clique from which the stabilisation propagates throughout the network. Here, $G_1 \in \mathcal{G}(4, 3, 1)$ and $G_2 \in \mathcal{G}(5, 3, 1)$. The partially illustrated graph $G_3 \in \mathcal{G}(4,3,k)$ is a cycle where there are additional edges to all neighbours within distance 3.}\label{fig:topologies}
\end{figure} 

\begin{lemma}
 Assume $\vec A$ is an algorithm that solves synchronous 2-counting in a complete network of $n$ nodes, out of which at most $f$ are faulty, with stabilisation time $t$ and with $s$ states per node. Then for any $G \in \mathcal{G}(n,2f+1,d)$, we can design an algorithm $\vec B$ that solves the synchronous 2-counting in $G$ using $s$ states per node. Moreover, $\vec B$ tolerates $f$ failures and stabilises in time $t+d-1$.
\end{lemma}
\begin{proof}
Let $G \in \mathcal{G}(n,2f+1,d)$ be our network topology. Fix a partition $V_0, \dots, V_d$ where $V_0 = \{1, \dots, n\}$ is a $n$-clique. We construct an algorithm $\vec B$ using the following rules:
 \begin{enumerate}
  \item If $i \in V_0 = K$, then $i$ outputs $A_i(x_1, \dots, x_n)$.
  \item If $i \in V_a$ for some $a > 0$, then node $i$ follows the majority of neighbours in $V_0 \cup \cdots \cup V_{a-1}$. If the majority is has output $y$, then output $1-y$. Otherwise output the current state.
 \end{enumerate}
We argue that at time step $t+r$, all nodes in $V_0 \cup \cdots \cup V_{r+1}$ have stabilised. The case of $r=0$ follows from Lemma~\ref{lemma:increase}. Suppose the claim holds for some $r'$ and consider node $i \in V_{r'+2}$. By the induction assumption and definition of $G$, $i$ has a set $P \subseteq V_0 \cup \cdots \cup V_{r'+1}$ of at least $2f+1$ neighbours.

Now node $i$ sees a majority of more than $f+1$ nodes in $P$ having the same output~$y$. Thus node $i$ outputs $1-y$ and is in agreement with non-faulty nodes in $P$ in the next round. Since there are $d+1$ sets in the partition of $V$, the algorithm stabilises in $t+d-1$ steps.
\end{proof}

It is known that consensus cannot be solved in networks with vertex-connectivity less than $2f+1$~\cite{dolev82byzantine}, and by Lemma~\ref{lemma:reduction}, this result carries over to synchronous 2-counting.

\paragraph{Beyond Synchronous Counting.}

We note that the previous lemmas hold for a larger class of problems as well: if it suffices that a node $v$ simply follows a majority of its neighbours, the generalisation techniques can be applied. These problems include, for example, binary consensus and set agreement~\cite{braud-santoni14synthesising}.

\section{Computer-Designed Algorithms}\label{sec:machine-pos}

In principle, we could now attempt to use a computer to tackle our original problem. By the discussion of \sectionref{sec:smalln}, it suffices to discover an algorithm with the smallest possible $s$ for the special case of $n = 4$ and $f = 1$. We could try increasing values of $s = 2, 3, \dotsc$. Once we have fixed $n$, $f$, and $s$, the problem becomes finite: an algorithm is a lookup table with $\ell = n s^n$ entries, and hence there are $s^\ell$ candidate algorithms to explore. For each candidate algorithm, we could use the projection graph approach of \sectionref{sec:projgraph} to quickly reject any invalid algorithm.

Unfortunately, the search space grows very rapidly and super-exponentially in the parameters $n$, $s$, and $f$. As we will see, there is no algorithm with $n = 4$ and $s = 2$. For $n = 4$ and $s = 3$, we have approximately $10^{154}$ candidates. We use three complementary approaches to tackle the task.

\begin{enumerate}
    \item Reduce (encode) the problem directly to propositional satisfiability and apply SAT solvers.
    \item Instead of directly encoding the problem as SAT, apply a SAT-based iterative counter-example guided abstraction refinement approach, in hope of better coping with the inherent combinatorial explosion.
    \item Narrow down the search space by also considering restricted classes of algorithms.
\end{enumerate}
The first approach is discussed in Section~\ref{sec:sat-synth} and the second approach in Section~\ref{sec:ceg-synth}. We will now describe the third approach, restricting the class of algorithms.

\paragraph{Cyclic Algorithms.}

We will consider two classes of algorithms---general algorithms (without any restrictions) and \emph{cyclic} algorithms. We say that algorithm $\vec A$ is cyclic if
\[
    A_i((x_i, x_{i+1}, \dotsc x_{n-1}, x_0, x_1, \dotsc, x_{i-1})) = A_0((x_0, x_1, \dotsc, x_{n-1}))
\]
for all $i$ and all $\vec x$. That is, a cyclic algorithm is invariant under cyclic renaming of the nodes.

There is no a priori reason to expect that the most efficient algorithms are cyclic. However, cyclic algorithms have many attractive features: for example, in a hardware implementation of a cyclic algorithm we only need to take $n$ copies of identical modules. Furthermore, the search space is considerably smaller: we only need to define transition function $A_0$. For $n = 4$ and $s = 3$, we have approximately $10^{38}$ candidate algorithms.

Cyclic algorithms are also much easier to verify. The projection graphs $G_F(\vec A)$ are isomorphic for all $|F| = 1$ and hence it is sufficient to check one of them.

\paragraph{Results.} 

We now present our main results on the new computer-generated algorithms and refer the discussion on how the results were obtained to Sections~\ref{sec:sat-synth}~and~\ref{sec:ceg-synth}.

The positive results are reported in \tableref{table:algorithms}. The key findings are a cyclic algorithm for $s=3$, $n=4$, and $f=1$, and a non-cyclic algorithm for $s=2$, $n=6$, and $f=1$. The table also gives examples of space-time tradeoffs: we can often obtain faster stabilisation if we use a larger number of states.

For the sake of comparison, we note that the \emph{fastest} deterministic algorithm from prior work~\cite{dolev07actions} stabilises in time $t = 13$ for $f = 1$ and it requires a large state space. Our algorithms achieve the stabilisation time of $t = 5$ for $s = 4$ and $t = 7$ for $s = 3$.

Machine-readable versions of all positive results, together with a Python script that can be used to verify the correctness of the algorithms, are freely available online~\cite{results}. Selected examples of the algorithms are also given in \appendixref{app:alg}. We also provide a compact, computer-checkable proof that shows that there is no algorithm for $s = 2$, $n = 4$, and $f = 1$, together with a verification program~\cite{results}.

\begin{table}
\center
\begin{tabular}{@{}l@{\qquad}c@{\qquad}c@{\qquad}c@{}}
  \toprule
  class & nodes ($n$) & states ($s$) & stabilisation time ($t$) \\
  \midrule
  cyclic & 4 & 3 & 7 \\
        & 5 & 3 & 6 \\
        & 6 & 3 & 3 \\
        & 7 & 2 & 8 \\
        & 8 & 2 & 4 \\
  \midrule
  general & 4 & 4 & 5 \\
          & 5 & 3 & 4 \\
          & 6 & 2 & 6 \\
  \bottomrule
\end{tabular}
\caption{Summary of computer-designed algorithms. The number of nodes $n$ is the smallest network on which the algorithm works and $t$ is the worst-case stabilisation time. }\label{table:algorithms}
\end{table}

\section{Synthesis via Directly Encoding to SAT}\label{sec:sat-synth}

In this section, we describe how to directly encode the synthesis problem into SAT. At a high level, we take the following approach:
\begin{enumerate}
    \item Fix the parameters $s$, $n$, $f$, $t$, and the algorithm family (cyclic or general).
    \item Construct a propositional formula $\varphi$ that is satisfiable iff an algorithm $\vec A$ for the given parameters exists.
    \item Use SAT solvers to find a satisfying assignment $\vec a$ of $\varphi$.
    \item Translate $\vec a$ to an algorithm $\vec A$.
\end{enumerate}

In essence, the formula $\varphi$ encodes the conditions given in \lemmaref{lem:projt} and the SAT solver (implicitly) searches through all algorithms $\vec A$:
\begin{enumerate}
    \item Guess an algorithm $\vec A$ and construct the projection graph $G_F(\vec A)$.
    \item Verify that there are no self-loops in $G_F$.
    \item Verify that the only successor of $\vec 0_F$ is $\vec 1_F$ and vice versa.
    \item For each $d = 1,2,\dotsc,t$, find the subset $B_F(d) \subseteq V_F$ of configurations with the following property: for each $\vec x \in B_F(d)$ there is a directed walk of length $d$ in $G_F$ that starts from $\vec x$ and does not traverse $\vec 0_F$ or $\vec 1_F$. We say that $\vec x \in B_F(d)$ is a \emph{$d$-bad} configuration.
    \item Verify that the set $B_F(t)$ is empty.
\end{enumerate}
For cyclic algorithms, we identify equivalent transitions and add corresponding equivalence constraints into the formula.

In the following, we describe the encoding by giving constraints for a single set $F \subseteq [n]$ of faulty nodes. The final formula is then the conjunction of these constraints over every possible choice of faulty nodes $F$.

\paragraph{Variables.}

Fix $F \subseteq [n]$ and let $\vec u \in [s]^n$, $\vec x, \vec y \in V_F$, $i \in [n]$, $d \in [t]$, and $c \in [s]$. We will use the following variables in the encoding:
\begin{itemize}[noitemsep]
    \item $a({\vec u, i, c})$ is true if $A_i(\vec u) = c$,
    \item $h({\vec x, i, c})$ is true if the adversary can force node $i$ to switch to state $c$ from configuration~$\vec x$,
    \item $e({\vec x, \vec y})$ is true if there exists an edge $(\vec x, \vec y) \in R_F$,
    \item $b({\vec x, d})$ is true if the configuration $\vec x \in B_F(d)$. 
\end{itemize}

\paragraph{Transition Functions.}

The $a$-variables describe the algorithm, that is, the transition function $A_i$ for each node $i$. Since we want each $A_i$ to be a well-defined function, we enforce the following constraints for all $\vec u \in V_F$, $i \in [n]$:
\begin{equation}
 \bigvee_{c \in [s]} a({\vec u, i, c})
\end{equation}
and, for all $c \in [s]$,
\begin{equation}\label{eq:deterministic}
 a({\vec u, i, c}) \rightarrow \big( \bigwedge_{ c' \in [s] \setminus c} \neg a({\vec u, i, c'}) \big).
\end{equation}
Observe that if the constraints given in (\ref{eq:deterministic}) are omitted, then $A_i$ may be a relation: a node may have several possible state transitions from a given observed state. Although one could always post-process each $A_i$ into a function, allowing transition relations instead of function will only help the adversary.

\paragraph{Projections.}

Let $\vec x, \vec y \in V_F$ be configurations. Recall from Section~\ref{sec:model} the definition of reachability. If the actual configuration is $\vec x$, then the adversary can choose any observed configuration from the set
$$
 U(\vec x) = \{ \vec u \in [s]^n \colon \pi_F(\vec u) = \vec x\}
$$
for each non-faulty node. For all $\vec u \in U(\vec x)$, we have
\begin{equation}
 a({\vec u,i,c}) \rightarrow h({\vec x, i, c}),
\end{equation}
declaring that the adversary can force node $i$ to switch to state $c$ from configuration $\vec x$. Now, the $h$-variables imply edges in the projection graph $G_F$:
\begin{equation}
 \bigwedge_{i \in [n] \setminus F} h({\vec x, i, y_i}) \rightarrow e({\vec x, \vec y}).
\end{equation}

\paragraph{Ensuring Counting Behaviour.}

The goal of the algorithm is to eventually stabilise and start oscillating only between the two actual configurations $\vec 0_F$ and $\vec 1_F$. To enforce this, we have the clauses
\begin{equation}
 e({\vec 0_F, \vec 1_F}) \text{ and } e({\vec 1_F, \vec 0_F}) 
\end{equation}
together with
\begin{equation}
 \neg e({\vec 0_F, \vec x})  \text{ and }  \neg e({\vec 1_F, \vec x})
\end{equation}
for all $\vec x \in V_F \setminus \{\vec 0_F, \vec 1_F\}$.

\paragraph{Forbidding Non-Stabilising Walks.}

First, we forbid self-loops in the projection graphs with the unit clause 
\begin{equation}
 \neg e({\vec x, \vec x}) 
\end{equation}
for every $\vec x \in V_F$. To ensure that all configurations but $\vec 0_F$ and $\vec 1_F$ belong to the set $B_F(0)$, we have the clauses
\begin{equation}
 \neg b({\vec 0_F, 0}) \text{ and }  \neg b({\vec 1_F,0}),
\end{equation} 
and, for each $\vec x \in V_F \setminus \{ \vec 0_F, \vec 1_F \}$, the clause
\begin{equation}
 b({\vec x, 0}).
\end{equation}
If a configuration $\vec x$ can reach a $d$-bad configuration $\vec y \in B_F(d)$, then $\vec x$ must be $(d+1)$-bad. This is captured by the clause
\begin{equation}
  \big( e({\vec x, \vec y})  \wedge b({\vec y, d}) \big)  \rightarrow b({\vec x,d+1}) 
\end{equation}
for each $\vec x, \vec y \in V_F$.
Finally, in order for the algorithm to eventually stabilise in the time limit $t$, we require that there are no $t$-bad configurations: 
\begin{equation}
 \neg b({\vec x, t}).
\end{equation}

\paragraph{Extension: Non-Uniform Stabilisation Time.}

It is straightforward to generalise the approach to non-uniform stabilisation time as follows, for some $t_0 < t$:
\begin{itemize}[noitemsep]
    \item if $|F| = 0$, the algorithm stabilises in time $t_0$,
    \item if $|F| = 1$, the algorithm stabilises in time $t$.
\end{itemize}
This can potentially help with the synthesis, by making the search space smaller, and it also helps with the quality of the algorithms.

Many of our algorithms are synthesised with this kind of encoding, with $t_0 = 2$ or $t_0 = 3$. Hence they not only work correctly in the presence of a Byzantine failure, but they also stabilise very quickly if all nodes are non-faulty. See the online supplement~\cite{results} for details.

\section{SAT-Based Counter-Example Guided Search}\label{sec:ceg-synth}

We now describe an alternative approach for synthesising synchronous counting algorithms: a counter-example guided search algorithm. The structure of our algorithm is similar to counter-example guided abstraction refinement techniques for model checking~\cite{clarke03cegar,clarke04sat-cegar} which have previously been successfully applied in various other computationally hard problem domains~\cite{janota12qbf,janota11abstraction,janota10circumscription,wintersteiger12qbf,demoura02lazy,barret02checking,flanagan03proving,dvorak14argumentation,finkbeiner12lazy}. We repeatedly (1)~try to construct an algorithm, (2)~check whether the algorithm is correct, and (3)~if not, then refine the encoding.

On a high-level, the search algorithm tries to guess a synchronous counting algorithm $\vec A$ and then uses a SAT solver to find a \emph{counter-example}, an execution that does not stabilise, for $\vec A$. If one is found, then the counter-example is used to include additional constraints to prune the search space, that is, to rule out at least the found counter-example from the implicit set of remaining algorithm candidates. Otherwise, $\vec A$ must be a correct algorithm.

\subsection{Encoding} 

For this approach, we use a symbolic encoding reminiscent of SAT-encodings for bounded model checking~\cite{biere09bmc}. As we want the SAT solver to verify that no counter-examples exist, we use an encoding where the SAT solver finds (i)~a set $F$ of faulty nodes and (ii)~a bad execution under $F$ for the counting algorithm.

\paragraph{Variables.}

Unlike previously, here we use a bit-wise encoding for the states. Each node has $B = \log(s) $ bits that represent its state. Here an observed configuration $\vec u$ is represented as a bit string of length $nB$; each node has $B$ bits to encode its state in $[s]$. If $s$ is not a power of two, then we add extra constraints that only allow $s$ states to be used.

We now list the variables used in the encoding and their semantics:
\begin{itemize}[itemsep=1pt] % LAYOUT
\item $p(i)$ is true if node $i$ is faulty. In other words, $p(i) = 1$ implies $i \in F$. 
\item $a({\vec u, i, b})$ represents the $b$th bit of the next state of node $i$ when it observes the configuration $\vec u \in \{0,1\}^{nB}$.
\item $u({i, j, b, k})$ is the $b$th bit of node $i$ as observed by node $j$ at timepoint $k$.
\item $z({k})$ and $o({k})$ are true if all non-faulty nodes are in state $0$ or $1$, respectively, at timepoint~$k$. 
\item $z(i,k)$ and $o(i,k)$ are true if $i$ is faulty \emph{or} in state $0$ or $1$, respectively, at timepoint~$k$.
\end{itemize}
We will also use the short-hand $g({i, b, k}) = u(i, i, b, k)$ to represent the $b$th bit of node $i$ at timepoint $k$. Next we define each part of the encoding as a separate formula.

\paragraph{Choosing the Set of Faulty Nodes.}

We now define the subformula $\psi_\text{faulty}$. We want the solver to be able to guess a set $F$ of faulty nodes under which a counter-example exists. To achieve this, we add constraints that force \emph{exactly} $f$ of the $p(i)$ variables to be true. 

In the following let $k \in [f]$, $i \in [n]$ and $j \in [n] \setminus \{0\}$. We will introduce the following variables:
\begin{itemize}[noitemsep]
    \item $p_=(k,i)$ is true if the $k$th faulty node is $i$.
    \item $p_\le(k,i)$ is true if the $k$th faulty node is at most $i$.
\end{itemize}

To enforce the semantics of these variables, we use the following clauses:
\begin{align}
 p_=(k,i) &\rightarrow p_\le(k,i), \\
 p_\le(k,j) &\rightarrow p_=(k,j) \vee p_\le(k,j-1), \\
 p_=(k,j-1) &\rightarrow \big( p_\le(k,j) \wedge \neg p_\le(k, j)), \\
 \neg p_=(k,0) &\rightarrow \neg p_\le(k,0).
\end{align}

To ensure that exactly $f$ faulty nodes will be chosen, we use the following clauses: we enforce that at least one node is designated as the $k$th faulty node with
\begin{equation}
 p_\le(k, n-1),
\end{equation}
and we enforce that there is a strict ordering among the nodes with 
\begin{equation}
 p_=(h-1, i) \rightarrow \neg p_\le(h,i)
\end{equation}
for all $h \in [f] \setminus \{0\}$. Finally, we establish the correspondence to $p(i)$ variables by enforcing
\begin{align}
 \big( p_=(k, i) \rightarrow p(i)\big) \wedge \big( p(i) \rightarrow \bigvee_{k' \in [f]} p_=(k',i) \big).
\end{align}

\paragraph{Trivial Transitions.}

Next, we give clauses that fix the trivial transitions for synchronous counting. The conjunction of these clauses is denoted as $\psi_\text{trivial}$.

Let $\vec 0$ and $\vec 1$ correspond to the observed configuration where all nodes are in state 0 or state 1, respectively. The state $0 \in [s]$ is encoded by a bit-string with all zeros, whereas $1 \in [s]$ is encoded as the 0th bit set to one and all other bits zero. 
Now, for all $i \in [n]$ and $b \in [B] \setminus \{0\}$, we enforce
\begin{equation}
 a(\vec 0, i, 0) \quad \text{ and } \quad \neg a(\vec 0, i, b),
\end{equation}
declaring that after observing configuration $\vec 0$, node $i$ must change its state to $1 \in [s]$. Conversely, from configuration $\vec 1$ we need to transition to state $0$. Thus, for all $b \in [B]$ we have
\begin{equation}
 \neg a(\vec 1, i, b).
\end{equation}

\paragraph{Representing State Transitions.}

Let $k \in [t]$. We now define the subformula $\psi_{k, \text{state}}$ encoding the systems behaviour at time step $k$.

If node $i$ is non-faulty, then the state of node $i$ is observed correctly by all other nodes. This is enforced with
\begin{equation}
 \neg p(i) \rightarrow \big( u(i, j, b, k) \leftrightarrow g(i,b,k) \big)
\end{equation}
for all $i,j \in [n]$ and $b \in [B]$.

For every observable configuration $\vec w \in [s]^n$, 
we introduce an auxiliary variable $d(\vec w, i, k)$ representing that node $i$ observes $\vec w$ at timepoint $k$.
Let $w(i, b)$ denote the $b$th bit of the binary representation of the state of node $i$ in the observed configuration~$\vec w$. 

To enforce the semantics of $d(\vec w, i, k)$,  for every observable configuration $\vec w \in [s]^n$ and every $j \in [n]$ the following constraint is needed:
\begin{equation}\label{eq:observed}
\begin{split}
\lnot d(\vec w, j, k)  \rightarrow
\Biggl (
&\Biggl ( \, \bigvee_{i \in [n],\, b \in [B] \colon w(i,b)=0} u(i, j, b, k) \Biggr ) \, \vee \\
&\Biggl ( \, \bigvee_{i \in [n],\, b \in [B] \colon w(i,b)=1} \lnot u(i, j, b, k) \Biggr ) 
\Biggr )
\end{split}
\end{equation}

The intuition behind \eqref{eq:observed} is that, if $d(\vec w, j, k)$ is false, 
then there must be at least one bit in the bit representation of the state observed by node $j$ at timepoint $k$ that is unequal to the bit representation of $\vec w$.

Finally, the state transitions of the system are enforced by the clauses
\begin{equation}\label{eq:transition}
d(\vec w, i, k-1) \rightarrow \big( g(i, b, k) \leftrightarrow a(\vec w, i, b) \big),
\end{equation}
where $k>0$, $\vec w \in [s]^n$, $i \in [n]$ and $b \in [B]$.
Equation \ref{eq:transition} enforces that if at the previous timepoint we observed state $\vec w$, 
then the state of node $i$ equals the successor state of $\vec w$ as specified by the transition relation of node $i$.

\paragraph{Indicators for Stabilisation.}

Finally, we define the behaviour of the $z$- and $o$-variables; the conjunction of these clauses is the subformula $\psi_{k, \text{indicator}}$. Recall that at timepoint $k$, the variable $z(k)$ is true iff the actual configuration is $\vec 0_F$, and respectively $o(k)$ is true iff the actual configuration is $\vec 1_F$. 
The equivalence is given by clauses which enforce for all $i \in [n], k \in [t]$:
\begin{equation}
 z(k) \rightarrow z(i,k) \quad \text{ and } \quad o(k) \rightarrow o(i,k),
\end{equation}
together with 
\begin{equation}
 \neg z(k) \rightarrow \bigvee_{j \in [n]} \neg z(j,k) \quad \text{ and } \quad \neg o(k) \rightarrow \bigvee_{j \in [n]} \neg o(j,k).
\end{equation}
It remains to describe the clauses that force the semantics of $z(i,k)$ and $o(i,k)$ variables. 
First, if a node $i$ is faulty then both $z(i,k)$ and $o(i,k)$ are forced to true:
\begin{align}
p(i) \rightarrow \big( z(i,k) \wedge o(i,k) \big).
\end{align}
For the $z$-variables, we enforce for all $b \in [B]$ the clauses
\begin{align}
z(i,k) \rightarrow \big( p(i) \vee \neg g(i,b, k) \big)
\end{align}
and the disjunction
\begin{equation}
\neg z(i,k) \rightarrow \bigvee_{b \in [B]} g(i, b, k),
\end{equation}
declaring that $z(i,k)$ is true iff $i$ is faulty or in state $0 \in [s]$. Similarly for the $o$-variables, as state $1 \in [s]$ was encoded as the bit string $10\dots 0$, we declare the following clauses to constrain the $o$-variables:
\begin{align}
o(i,k) &\rightarrow \Big( p(i) \vee \big( g(i, 0, k) \wedge \bigwedge_{b \in [B] \setminus \{0\}} \neg g(i, b, k) \big) \Big)
\end{align}
together with the disjunction
\begin{align}
\neg o(i,k) &\rightarrow \big( \neg g(i, 0, k) \vee \bigvee_{b \in [B] \setminus \{0\}} g(i, b, k) \big).
\end{align}

\paragraph{Combining the Subformulas.}

The counter-example guided search algorithm  incrementally builds a propositional formula to use for both verification and synthesis. In the algorithm description, we will refer to the following formulas:
\begin{align}
\psi_\text{base} = \psi_\text{faulty} \wedge \psi_\text{trivial},
\end{align}
which gives the basis of the encoding, and, for each $k \ge 0$,
\begin{align}
\tau_k = \psi_{k, \text{state}} \wedge \psi_{k, \text{indicator}},
\end{align}
which encodes the unrolling of time.

\subsection{Basic Search Algorithm}

Our search algorithm will iteratively construct a sequence $\Psi_0, \Psi_1, \dots$ of formulas until a stabilising 2-counting algorithm is found. Given a satisfiable formula $\Psi_i$, a satisfying assignment $\rho$ defines the following:
\begin{itemize}
 \item $\vec A(\rho)$: an algorithm defining the $n$ transition functions $A_1, \dots, A_n$,
 \item $F(\rho) \subseteq [n]$: a set of $f$ faulty nodes,
 \item $X(\rho) = ( \vec x^0, \dotsc, \vec x^{k} )$: an execution of $\vec A$ under the set $F(\rho)$ of faulty nodes,
 \item $U(\rho) = \{ \vec u_{ij} : i \in [n] \setminus F(\rho), j \in [k] \}$: the configurations observed by non-faulty nodes.
\end{itemize}
That is, the algorithm $\vec A(\rho)$ is determined by the $a(\cdot)$ variables assigned true in $\rho$, the set $F(\rho)$ by the $p(\cdot)$ variables, and so on.

If an assignment $\rho$ exists, then either $\vec A(\rho)$ is a correct algorithm or $X(\rho)$ gives an execution that violates the specification of synchronous 2-counting. In the latter case, the search algorithm inspects $X(\rho)$ and adds constraints that forbid any other solutions $\rho'$ such that $\vec A(\rho) = \vec A(\rho')$. Of course, a na\"ive approach is to add constraints that explicitly exclude algorithm $\vec A$. However, inspecting the transition functions carefully allows for more frugal constraints that forbid several algorithms, that is, a tighter abstraction refinement.

\begin{figure}
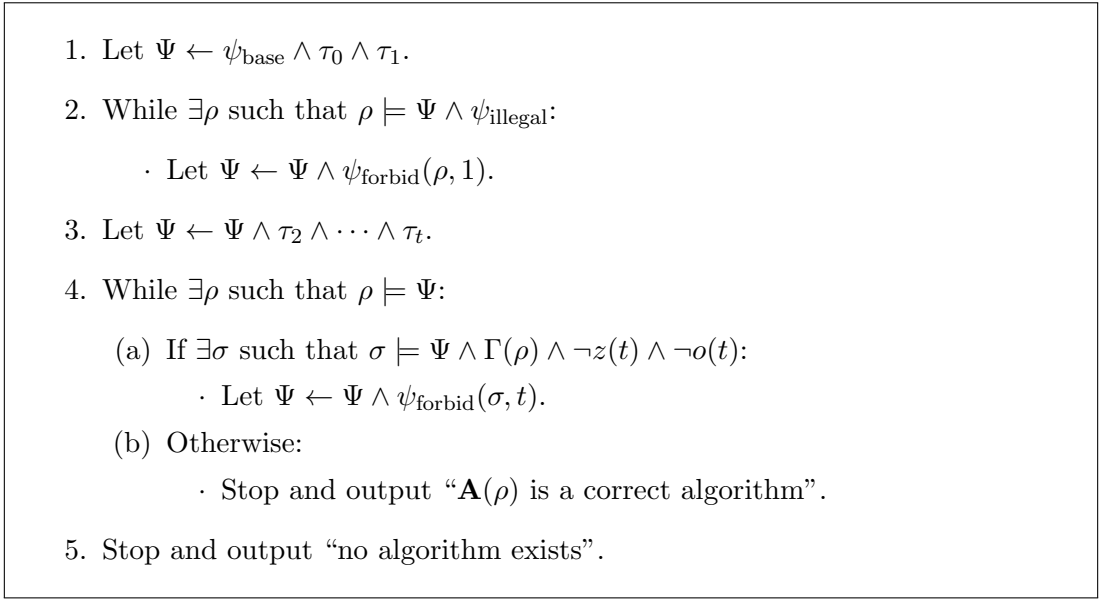

\algorithm{
    \item Let $\Psi \gets \psi_{\textrm{base}} \wedge \tau_0 \wedge \tau_1$.
    \item While $\exists \rho$ such that $\rho \models \Psi \wedge \psi_{\textrm{illegal}}$:
    \block{
        \item Let $\Psi \gets \Psi \wedge \psi_\textrm{forbid}(\rho, 1)$.
    }
    \item Let $\Psi \gets \Psi \wedge \tau_2 \wedge \cdots \wedge \tau_t$.
    \item While $\exists \rho$ such that $\rho \models \Psi$:
    \begin{enumerate}
        \item If $\exists \sigma$ such that
        $\sigma \models \Psi \wedge \Gamma(\rho) \wedge \neg z(t) \wedge \neg o(t)$:
        \block{
            \item Let $\Psi \gets \Psi \wedge \psi_\textrm{forbid}(\sigma, t)$.
        }
        \item Otherwise:
        \block{
            \item Stop and output ``$\vec A(\rho)$ is a correct algorithm''.
        }
    \end{enumerate}
    \item Stop and output ``no algorithm exists''.
    \vspace{-3mm}
}
\vspace{-4mm}
\caption{Basic search algorithm. Steps 2, 4, and 4a resort to a SAT solver to find a satisfying assignment of a given formula.}\label{fig:alg-basic}
\end{figure}

The \emph{basic search algorithm} is given in \figureref{fig:alg-basic}. Step 1 defines the initial formula that acts as a basis for the incremental search. In Step 2, the search algorithm first removes all algorithm candidates that do not correctly oscillate between the $\vec 0_F$ and $\vec 1_F$ states even in the special case when the system starts from either state. The formula $\psi_{\textrm{illegal}}$ is defined as $(z(0) \wedge \neg o(1)) \vee (o(0) \wedge \neg z(1))$, and the formulas $\psi_\textrm{forbid}(\cdot, \cdot)$ are constraints that remove bad algorithms from the search space---we will describe these in detail in \sectionref{ssec:refinement}.

Step 4 asks the SAT solver to guess an algorithm candidate $\vec A(\sigma)$. In Step 4a, the SAT solver is used to find a counter-example to $\vec A(\sigma)$ to see whether it stabilises. If a counter-example is found, then we use the counter-example to add more constraints to prune the search space. Here, the formula $\Gamma(\rho)$ encodes $\vec A(\rho)$ as a conjunction of literals consisting of variables $a(\vec u, i, b)$. Step 4b is reached if no counter-example is found, meaning that $\vec A$ is a correct algorithm for synchronous counting.

Finally, if we reach Step 5, we know that $\Psi$ is unsatisfiable, and hence, there does not exist any correct algorithms for the given parameters.  

\begin{remark}
Note that there exist several possible trade-offs between having a simple search algorithm and speeding up synthesis by introducing problem-specific knowledge into the algorithm and encoding. For example, Step~2 essentially learns \lemmaref{lem:projt}.1 which we could also directly encode into the base formulas. In Step 4, we can introduce $z(0)$ as a conjunct into the formula to make the search for $\vec A(\sigma)$ intuitively easier, and so on. However for clarity of exposition, we will focus on more general algorithmic ideas instead of problem-specific tunings.
\end{remark}

\subsection{Refinement through Counter-Examples}\label{ssec:refinement}

Once the SAT solver finds a counter-example, we need to forbid algorithms that exhibit the incorrect behaviour. Intuitively, we add constraints that force the change of \emph{some} transitions that caused the bad execution. 

Formally, we construct $\psi_\text{forbid}(\sigma, k)$ as follows. Let $\vec x^0, \dots, \vec x^k$ be the execution $X(\sigma)$ and let $\vec u_{ij}$ be the configuration observed by node $i \notin F(\sigma)$ at timepoint $j < k$. The literals responsible for the transitions are divided into two sets, $P^+$ and $P^-$, as follows:
\begin{align*}
 (i,j,b) \in P^+ \quad &\textrm{iff} \quad \sigma[a(\vec u_{ij}, i, b)] = 1 \\
 (i,j,b) \in P^- \quad &\textrm{iff} \quad \sigma[a(\vec u_{ij}, i, b)] = 0.
\end{align*}
Above, $\sigma[x] \in \{0,1,\bot\}$ denotes the value (false, true, unassigned) of variable $x$ in assignment $\sigma$. Now the constraint is 
\begin{equation}
 \psi_\text{forbid}(\sigma, k) = \bigvee_{(i,j,b) \in P^+} a(\vec u_{ij}, i, b) \vee  \bigvee_{(i,j,b) \in P^-} \neg a(\vec u_{ij}, i, b).
\end{equation}
Note that the case $P^+ = P^- = \emptyset$ must be a contradiction, and hence the formula is always non-empty.

\subsection{Improvement: Finding Short Loops}

The constraint can be strengthened when $X(\sigma)$ contains a loop $\vec x^0, \dots, \vec x^h$ for some $h < k$, by then only considering timepoints $j \leq h$ when constructing the sets $P^+$ and $P^-$. Then, instead of stating that some transition must be changed in the entire length-$k$ execution, we state that it suffices to  change something for only $h < k$ of the steps. This results in a shorter disjunction in the constraint.

To this end, we modify Step 4 in the basic search algorithm as shown in \figureref{fig:alg-basic2}. We introduce a new variable $\ell(k)$ which is true iff $\vec x^0 = \vec x^k$. We first find the smallest $k < t$ for $\vec A(\rho)$ such that a bad execution consisting of a length-$k$ loop exists. If no such loop exists, we proceed as before. Otherwise, we use the counter-example consisting of a loop to refine the current abstraction.

\begin{figure}
\begin{framed}
\begin{enumerate}[start=4]
    \item While $\exists \rho$ such that $\rho \models \Psi$:
    \begin{enumerate}
        \item If $\exists k \le t$ and $\exists \sigma$ such that
        $\sigma \models \Psi \wedge \Gamma(\rho) \wedge \ell(k)$:
        \block{
            \item Let $\Psi \gets \Psi \wedge \psi_\text{forbid}(\sigma, k^*)$, where $k^*$ is the smallest such $k$.
        }
        \item Otherwise, if $\exists \sigma$ such that
        $\sigma \models \Psi \wedge \Gamma(\rho) \wedge \neg z(t) \wedge \neg o(t)$:
        \block{
            \item Let $\Psi \gets \Psi \wedge \psi_\textrm{forbid}(\sigma, t)$.
        }
        \item Otherwise:
        \block {
            \item Stop and output ``$\vec A(\rho)$ is a correct algorithm''.
        }
    \end{enumerate}
\vspace{-5mm}
\end{enumerate}
\end{framed}
\vspace{-4mm}
\caption{Finding short loops: modifications to Step 4 of \figureref{fig:alg-basic}.}\label{fig:alg-basic2}
\end{figure}

\subsection{Improvement: Overshooting and Unrolling on Demand}

Usually we are interested in knowing whether there exist \emph{any} stabilising counting algorithm for given parameter values $s$, $n$, and $f$. For this task, we modify the search algorithm so that it can first quickly find some algorithm, possibly with a very long stabilisation time, and then gradually further tightening the stabilisation-time requirement.

\begin{figure}
\algorithm{
    \item Let $\Psi \gets \psi_{\textrm{base}} \wedge \tau_0 \wedge \tau_1$.
    \item While $\exists \rho$ such that $\rho \models \Psi \wedge \psi_{\textrm{illegal}}$:
    \block{
        \item Let $\Psi \gets \Psi \wedge \psi_\textrm{forbid}(\rho, 1)$.
    }
    \item Let $k \gets 1$. 
    \item While $\exists \rho$ such that $\rho \models \Psi \wedge z(0)$:
    \begin{enumerate}
        \item Let $i \gets \min \bigl\{j \le k : \exists \sigma_j \textrm{ such that } \sigma_j \models \Psi \wedge \Gamma(\rho) \wedge \ell(j) \bigr\} \cup \{\infty\}$.
        \item If $i \le k$:
        \block{
            \item Let $\Psi \gets \Psi \wedge \psi_\textrm{forbid}(\sigma_i, i)$.
        }
        \item Otherwise, if $\exists \pi \textrm{ such that } \pi \models \Psi \wedge \Gamma(\rho) \wedge \neg z(k) \wedge \neg o(k)$:
        \block{
            \item If $k < t$:
            \block{
                \item Let $k \gets k+1$ and $\Psi \gets \Psi \wedge \tau_k$,
                \item Resume from Step~4a.
            }
            \item Otherwise:
            \block{
                \item Let $\Psi \gets \Psi \wedge \psi_\textrm{forbid}(\pi, k)$.
            }
        }
        \item Otherwise:
        \block{
            \item Output ``$\vec A(\rho)$ is a correct algorithm that stabilises in $k$ steps'',
            \item Let $k \gets k-1$ and $t \gets k$,
            \item Resume from Step~4b.
        }
    \end{enumerate}
    \item Stop and output: ``no algorithm exists that stabilises in time $t$''.
    \vspace{-2mm}
}
\vspace{-4mm}
\caption{Overshooting algorithm.}\label{fig:alg-overshoot}
\end{figure}

The \emph{overshooting algorithm} is given in \figureref{fig:alg-overshoot}. It unrolls the encoding on demand. By setting $t = \infty$, the algorithm tries to find \emph{any} algorithm that stabilises. Of course, as the state space is finite, there is also a finite upper bound on $t$ that can be used here.

The algorithm works as follows. Step 4a searches for the smallest $i$ such that a $i$-loop counter-example exists for $\vec A(\rho)$. In Step 4b, if we have already unrolled the execution to at least $i$ steps, then we add new constraints. Otherwise, Step 4c attempts to find a counter-example $\pi$ of length $k$. If $k < t$, then we unroll the encoding for one additional time step, as it may be that our current time bound $k$ is too small for a stabilising algorithm to exist. Otherwise, we prune the search space using the counter-example $\pi$.

\section{Empirical Results}\label{sec:comparison-results}

So far we have introduced two different approaches for constructing synchronous counting algorithms. Now the obvious question remains: \emph{which one is better?} To answer this, we empirically compared the direct encoding given in \sectionref{sec:sat-synth} against the counter-example guided algorithm described in \sectionref{sec:ceg-synth}. In particular, our goal was to find out which method is more useful in practice when one wants to synthesise new algorithms.

\paragraph{Solvers.}

For solving instances via the direct propositional encoding, we used two freely available 
state-of-the-art complete SAT solvers: {\sc MiniSAT}~\cite{een04minisat} (version 2.2.0 with simplifications) and {\sc lingeling} (version \texttt{ayv})~\cite{biere14lingeling}. The input formula was encoded in the standard DIMACS CNF file format. As both solvers allow a wide range of different parameters to fine-tune the solver search routines, we settled on running both solver using their respective default parameters.

Our implementation of the counter-example guided search, dubbed as {\sc symsync}, builds on top of the incremental interface of the {\sc MiniSAT} solver~\cite{een04temporal}. We used the overshooting variant of the counter-example guided search. Thus, the solver relaxes the time bound when it does not find a correct algorithm matching the target stabilisation time, but after finding some stabilising algorithm, the solver will then gradually tighten the time bound.

\paragraph{Experiment Setup.}

Recall that an instance of the synthesis problem consists of the class of algorithms (general or cyclic) and four parameters: number of nodes $n$, faulty nodes $f$, states $s$, and the stabilisation time $t$. We chose a set of problem instances consisting of both realisable (an algorithm exists) and unrealisable (no algorithm exists) instances, as listed in Table~\ref{table:instances}. We attempted to choose instances of various difficulty, but still solvable within a four hour limit on CPU time; we note that some of the algorithms presented in Table~\ref{table:algorithms} of \sectionref{sec:machine-pos} required considerably longer time to synthesise.

\begin{table}
\center
\begin{tabular}{@{}l@{\qquad}c@{\qquad}c@{\qquad}c@{\qquad}c@{\qquad}c@{\qquad}c@{\qquad}c@{}}
  \toprule
  class & $n$ & $s$ & $t$ & realisable? & $\log_{10}$ of \#candidates \\
  \midrule
  cyclic & 4 & 3 & 6 & no & 38 \\
         & 7 & 2 & 3 & no & 38  \\ 
         & 8 & 2 & 3 & no & 77 \\
  \midrule
         & 4 & 3 & 7 & yes & 38 \\
         & 5 & 3 & 6 & yes & 115 \\
         & 6 & 3 & 3 & yes & 347 \\
         & 7 & 2 & 8 & yes & 38 \\
         & 8 & 2 & 4 & yes &77  \\
  \midrule
  general & 4 & 3 & 7 & yes & 154 \\
          & 5 & 2 & 79 & no & 48 \\
          & 5 & 3 & 4 & yes & 579 \\
          & 6 & 2 & 6 & yes & 115 \\
          & 7 & 2 & 8 & yes & 269 \\
  \bottomrule
\end{tabular}
\caption{Problem instances used in the empirical experiments. For all realisable instances, we also run the experiments for relaxed instances with stabilisation time $t+1$, $2t$, and the maximum possible stabilisation time. The last column gives the $\log_{10}$ of the number of algorithm candidates. }\label{table:instances}
\end{table}

For each problem instance, we ran $N=100$ copies of each of the three solvers, initialising every process with a different random seed. We recorded the running time, the maximum memory footprint, and other statistics for each process. When using the direct encoding, we did not include the time required to generate the instance. The experiments were run on a computing cluster with Intel Xeon X5650 2.67-GHz processors. Each process was single-threaded and the memory limit was set to 8 GB. 

For each realisable problem instance listed in Table~\ref{table:instances}, we also ran the same experiment setup as above for relaxed instances by increasing the stabilisation time bound in three ways: increasing the stabilisation bound by one, doubling the bound, and finally using the maximal bound of $t=s^{n-f}-2$ time steps. Intuitively, suboptimal algorithms with a longer stabilisation time should be more common, and hence, perhaps easier to find. However, this also increases the size of the search space and the size of the SAT instances.

\begin{table} 
\center
\begin{tabular}{@{}l@{\quad}c@{\enspace}c@{\enspace}r@{\qquad}r@{\enspace}r@{\enspace}r@{\qquad}r@{\enspace}r@{\enspace}r@{\qquad}r@{\enspace}r@{\enspace}r@{}}
  \toprule
  \multicolumn{4}{@{}c@{\qquad}}{} &
  \multicolumn{9}{@{}c@{}}{Running time (seconds)} \\
  \cmidrule{5-13}
  \multicolumn{4}{@{}c@{\qquad}}{Instance} &
  \multicolumn{3}{@{}c@{\qquad}}{\sc MiniSAT} &
  \multicolumn{3}{@{}c@{\qquad}}{\sc lingeling} &
  \multicolumn{3}{@{}c@{}}{\sc symsync} \\
  \cmidrule(r{2em}){1-4}
  \cmidrule(r{2em}){5-7}
  \cmidrule(r{2em}){8-10}
  \cmidrule{11-13}
  class & $n$ & $s$ & $t$ & 10\% & 50\% & 90\% & 10\% & 50\% & 90\% & 10\% & 50\% & 90\% \\
  \midrule
cyclic & 4 & 3 &7 & 1 & 1 & 1 & 1 & 1 & 1 & 1 & 2 & 6 \\
 & & &8 & 1 & 1 & 3 & 1 & 1 & 1 & 1 & 1 & 5 \\
 & & &14 & 1 & 1 & 1 & 1 & 1 & 1 & 1 & 1 & 4 \\
 & & &25 & 1 & 1 & 2 & 1 & 1 & 1 & 1 & 1 & 4 \\
\cmidrule{2-13}
 & 5 & 3 &6 & 2373 & --- & --- & \bftab 803 & 2715 & --- & --- & --- & --- \\
 & & &7 & 1477 & 13305 & --- & \bftab 44 & 632 & 711 & $\star$ & --- & --- \\
 & & &12 & 25 & 436 & 3009 & 12 & 16 & 91 & \bftab 5 & 31 & 1014 \\
 & & &79 & 66 & 672 & 4180 & 114 & 167 & 441 & \bftab 3 & 18 & 42 \\
\cmidrule{2-13}
 & 6 & 3 &3 & 79 & 3634 & --- & \bftab 16 & 22 & 70 & --- & --- & --- \\
 & & &4 & $\star$ & --- & --- & \bftab 178 & 272 & 3734 & --- & --- & --- \\
 & & &6 & 2053 & --- & --- & \bftab 251 & 2451 & 4344 & $\star$ & --- & --- \\
 & & &241 & 6930 & --- & --- & 1981 & 2735 & --- & \bftab 41 & 505 & --- \\
\cmidrule{2-13}
 & 7 & 2 &8 & \bftab 34 & 604 & 4177 & 65 & --- & --- & $\star$ & --- & --- \\
 & & &9 & 32 & 560 & 2356 & \bftab 21 & 26 & 101 & 5233 & --- & --- \\
 & & &16 & 16 & 102 & 661 & 18 & 72 & 79 & \bftab 2 & 20 & 84 \\
 & & &62 & 41 & 442 & 1921 & 60 & 185 & 267 & \bftab 2 & 5 & 35 \\
\cmidrule{2-13}
 & 8 & 2 &4 & \bftab 7 & 101 & 440 & 19 & 67 & 81 & --- & --- & --- \\
 & & &5 & \bftab 15 & 119 & 797 & 28 & 56 & 83 & --- & --- & --- \\
 & & &8 & 62 & 558 & 3000 & \bftab 50 & 56 & 216 & 622 & 7304 & --- \\
 & & &126 & 850 & 4117 & --- & 967 & 3945 & 7993 & \bftab 9 & 21 & 145 \\
\midrule
general & 4 & 3 &7 & 10859 & --- & --- & \bftab 4246 & --- & --- & --- & --- & --- \\
 & & &8 & 2639 & --- & --- & \bftab 497 & --- & --- & --- & --- & --- \\
 & & &14 & \bftab 2884 & --- & --- & 3211 & --- & --- & --- & --- & --- \\
 & & &25 & \bftab 2600 & --- & --- & 13639 & --- & --- & --- & --- & --- \\
\cmidrule{2-13}
 & 5 & 3 &4 & $\star$ & --- & --- & $\star$ & --- & --- & --- & --- & --- \\
 & & &5 & $\star$ & --- & --- & $\star$ & --- & --- & --- & --- & --- \\
 & & &8 & $\star$ & --- & --- & $\star$ & --- & --- & --- & --- & --- \\
 & & &79 & --- & --- & --- & --- & --- & --- & --- & --- & --- \\
\cmidrule{2-13}
 & 6 & 2 &6 & --- & --- & --- & --- & --- & --- & 1167 & --- & --- \\
 & & &7 & --- & --- & --- & --- & --- & --- & 541 & --- & --- \\
 & & &12 & $\star$ & --- & --- & $\star$ & --- & --- & 69 & 1782 & --- \\
 & & &30 & $\star$ & --- & --- & $\star$ & --- & --- & 46 & 382 & 2069 \\
\cmidrule{2-13}
 & 7 & 2 &8 & --- & --- & --- & --- & --- & --- & 528 & --- & --- \\
 & & &9 & --- & --- & --- & --- & --- & --- & 354 & 8990 & --- \\
 & & &16 & --- & --- & --- & --- & --- & --- & 111 & 946 & --- \\
 & & &62 & --- & --- & --- & --- & --- & --- & 75 & 415 & --- \\
 \bottomrule
\end{tabular}
\caption{Summary of \emph{realisable} problem instances. The solver columns indicate the first, fifth (median), and ninth decile of running times in seconds. Columns marked with $\star$ indicate that a solution was found by some but less than 10\% of the processes. For the first decile we have highlighted the running time of the fastest solver. Here $f=1$ for all cases.}\label{table:realisable}
\end{table}

\begin{figure}
    \centering%
    \includegraphics{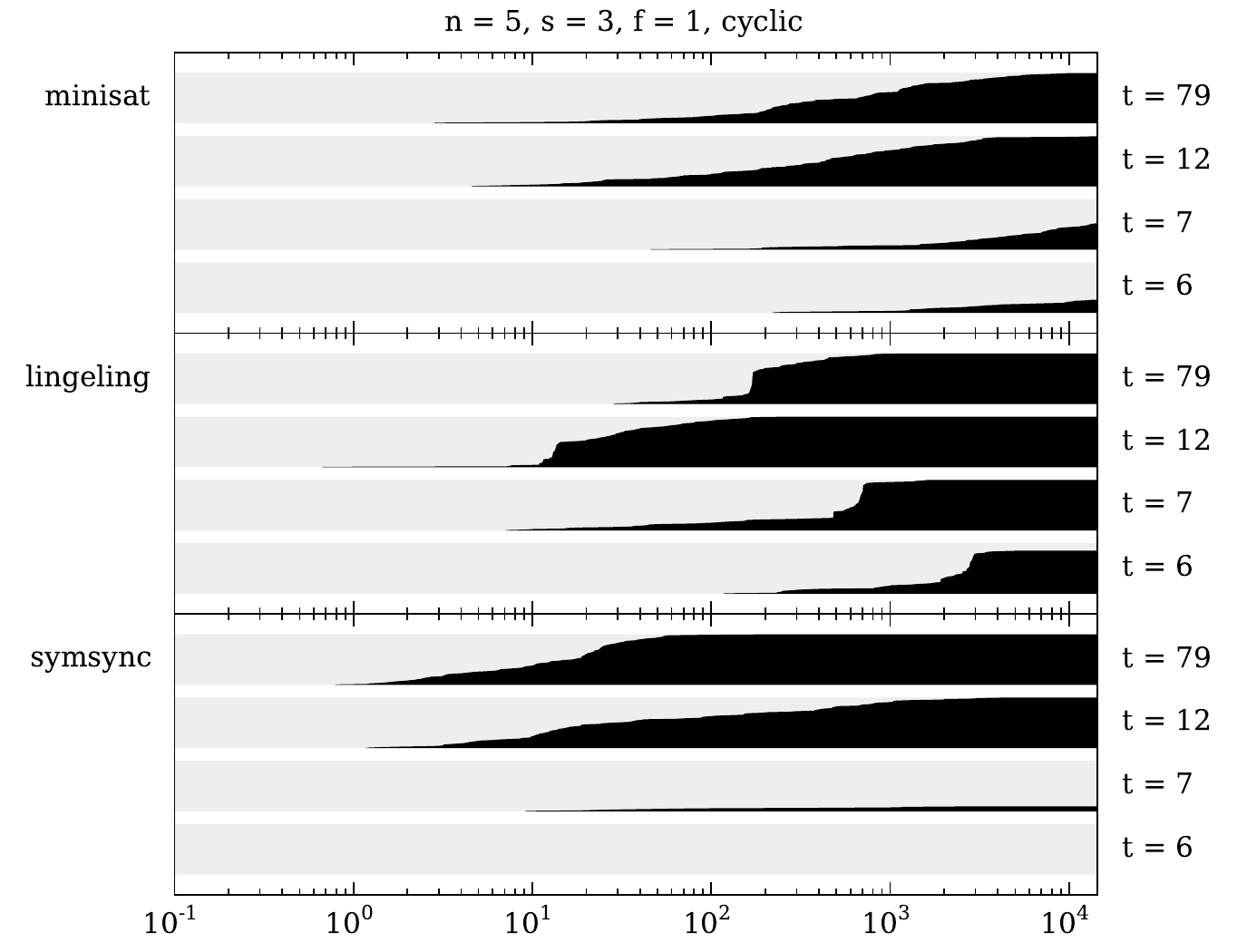}\\[5mm]%
    \includegraphics{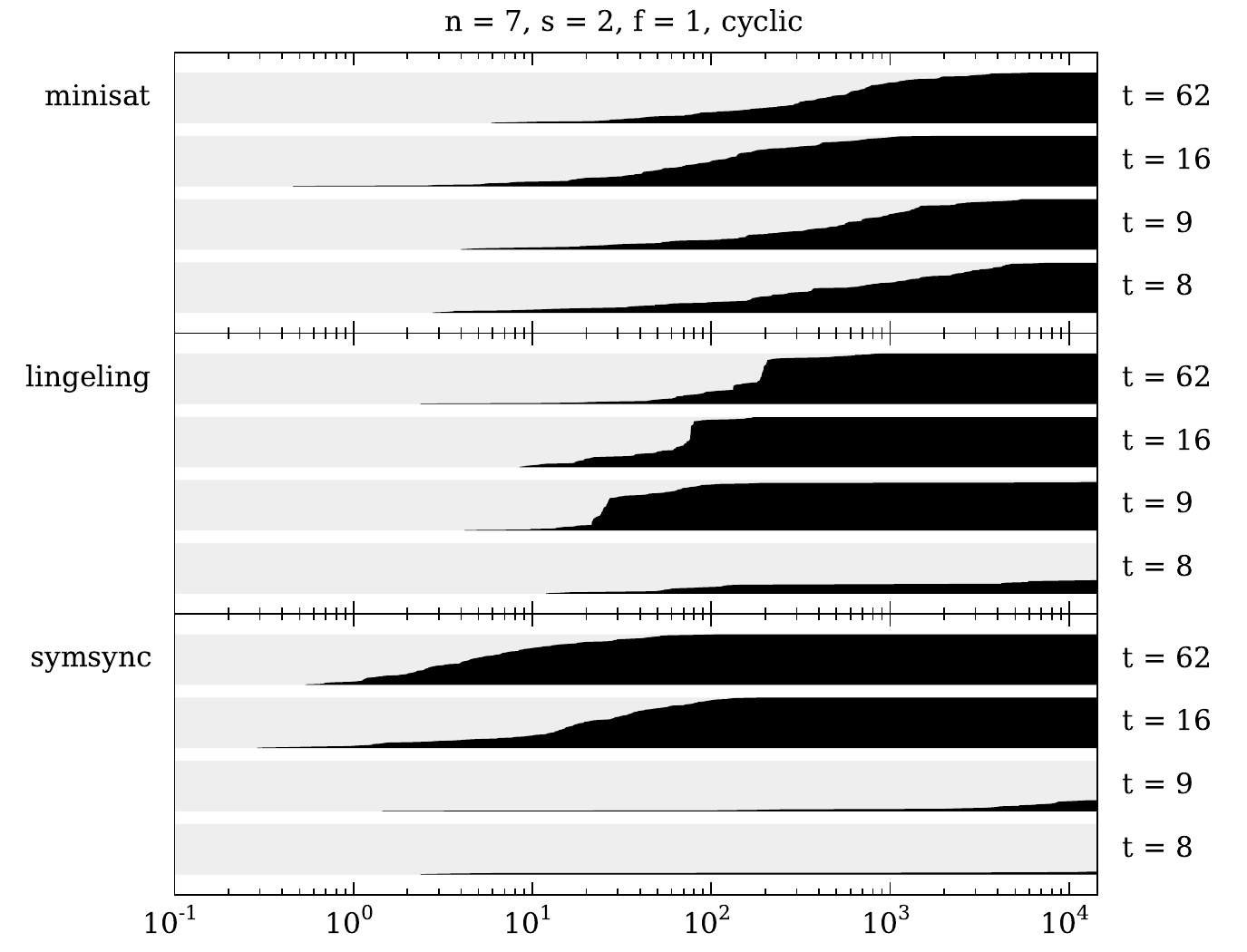}%
    \caption{Example of synthesis times. The $x$ axis is the logarithm of time in seconds and the $y$-axis is the fraction of processes that had solved the problem instance.}\label{fig:plots}
\end{figure}

\paragraph{Results.}

The synthesis times for realisable instance are summarised in Table~\ref{table:realisable} and Figure~\ref{fig:plots}. For each solver, the table gives the median together with first and ninth decile of synthesis times (in seconds). The time to generate the propositional formula for direct encoding instances is not included in the running times of {\sc MiniSAT} and {\sc lingeling} solvers, but is for {\sc symsync} solver, as it iteratively generates its internal encoding within the CEGAR loop  during execution.

The immediate observation is that neither direct encoding or the CEGAR approach consistently outperform the other. However, it is easy to see some patterns. First, the direct encoding works well for finding \emph{optimal} or nearly-optimal algorithms, but finding \emph{some} algorithm is much faster with {\sc symsync}. On the other hand, {\sc symsync} rarely manages to find optimal algorithms within the time limit of four hours or the memory limit of eight gigabytes.

Typically, when the solvers failed to find a solution, this was due to hitting the time limit. The only notable exceptions to this were the instances for general algorithms with $n=5$ and $s=3$, where each {\sc symsync} instance ran out of memory in each case, and the cyclic instances with $n=6$ and $s=3$, where most of the failures were caused by running out of memory. Neither {\sc MiniSAT} nor {\sc lingeling} ran out of memory in these experiments.

The second pattern is that in many cases {\sc symsync} gives solutions to instances with $s=2$ states at least an order of magnitude faster than the direct encoding approach. For general algorithms with $n \in \{6,7\}$, the direct encoding approach does not even produce results in the given time limit.

Indeed the observed behaviour is expected. The {\sc symsync} solver refines the abstraction and relaxes the time bound if a fast algorithm is not found steadily increasing the size of the encoding. Usually, some algorithm will be encountered, and from thereon, the solver will simply proceed by adding new constraints until an algorithm with the desired time bound is found. On the other hand, trying to find \emph{some} algorithm using the direct encoding amounts to simply increasing the time bound to a large enough value right from the start---this greatly increases the size of the propositional formula making the search slower.

When comparing the two different SAT solvers used in the direct encoding approach, rather unsurprisingly, the actively developed {\sc lingeling} solver outperforms {\sc MiniSAT}. We suspect that {\sc lingeling} greatly benefits from its inprocessing capabilities, which are not present in the other solvers.

The results for unrealisable instance are listed in Table~\ref{table:unrealisable}. For unrealisable instances, it is relatively clear that the direct encoding outperforms the counter-example guided approach, although {\sc symsync} is able to prove the non-existence of a two-state algorithm for $n=5$ nodes in time comparable to the direct encoding approach.

\begin{table}
\center
\begin{tabular}{@{}l@{\quad}c@{\enspace}c@{\enspace}r@{\qquad}r@{\enspace}r@{\enspace}r@{\qquad}r@{\enspace}r@{\enspace}r@{\qquad}r@{\enspace}r@{\enspace}r@{}}
  \toprule
  \multicolumn{4}{@{}c@{\qquad}}{} &
  \multicolumn{9}{@{}c@{}}{Running time (seconds)} \\
  \cmidrule{5-13}
  \multicolumn{4}{@{}c@{\qquad}}{Instance} &
  \multicolumn{3}{@{}c@{\qquad}}{\sc MiniSAT} &
  \multicolumn{3}{@{}c@{\qquad}}{\sc lingeling} &
  \multicolumn{3}{@{}c@{}}{\sc symsync} \\
  \cmidrule(r{2em}){1-4}
  \cmidrule(r{2em}){5-7}
  \cmidrule(r{2em}){8-10}
  \cmidrule{11-13}
  class & $n$ & $s$ & $t$ & 10\% & 50\% & 90\% & 10\% & 50\% & 90\% & 10\% & 50\% & 90 \% \\
  \midrule
cyclic & 4 & 3 &6 & 2 & 3 & 3 & 4 & 6 & 6 & --- & --- & --- \\
 & 7 & 2 &7 & --- & --- & --- & --- & --- & --- & --- & --- & --- \\
 & 8 & 2 &3 & 9405 & 13809 & --- & 999 & 1364 & 1612 & --- & --- & --- \\
\midrule
general & 5 & 2 &79 & 1148 & 1502 & 2016 & 1563 & 2353 & 2927 & 2482 & 2780 & 3421 \\
  \bottomrule
\end{tabular}
\caption{Summary of \emph{unrealisable} problem instances. }\label{table:unrealisable}
\end{table}

\section{Conclusions}

In this work, we have used computational techniques to study the synchronous counting problem. At first sight the problem is not well-suited for computational algorithm design---we need to reason about stabilisation from any given starting configuration, for any adversarial behaviour, in a system with arbitrarily many nodes. Nevertheless, we have demonstrated that computational techniques can be used in this context to discover novel algorithms.

Our algorithms outperform the best human-designed algorithms: they are deterministic, small ($2 \le s \le 3$), fast ($3 \le t \le 8$), and easy to implement in hardware or in software---a small lookup table suffices. In summary, our work leaves very little room for improvement in the case of $f = 1$. The general case of $f > 1$ is left for future work; we are optimistic that the algorithms designed in this work can be used as subroutines to construct algorithms that tolerate a larger number of failures.

We presented two complementary approaches for algorithm synthesis: the direct SAT encoding from \sectionref{sec:sat-synth} and the SAT-based CEGAR approach from \sectionref{sec:ceg-synth}. In our experiments, the direct encoding was typically the fastest method for finding \emph{optimal} algorithms, while the CEGAR approach quickly discovered \emph{some} algorithms.

Even though our computer-generated algorithms are constructed with a fairly complicated tool\-chain, the end results are compact, machine readable, and easy to verify with a straightforward script. All results and the verification tools are freely available online~\cite{results}.

\section*{Acknowledgements}

This work is an extended and revised version of a preliminary conference report~\cite{dolev13counting}. We thank Josef Widder and Igor Konnov for helpful suggestions, and Nicolas Braud-Santoni, Aristides Gionis, Tomi Janhunen, Jussi Rintanen and Ulrich Schmid for discussions.

DD: Danny Dolev is Incumbent of the Berthold Badler Chair in Computer Science. This research project was supported in part by The Israeli Centers of Research Excellence (I-CORE) program, (Center  No. 4/11), by grant 3/9778 of the Israeli Ministry of Science and Technology, and by the ISG (Israeli Smart Grid) Consortium, administered by the office of the Chief Scientist of the Israeli Ministry of Industry and Trade and Labor.

MJ: Work supported by Academy of Finland under grants 251170 COIN Centre of Excellence in Computational Inference Research, 276412, and 284591.

JHK, JR, JS: This work was supported in part by the Helsinki Doctoral Programme in Computer Science -- Advanced Computing and Intelligent Systems, by the Academy of Finland (grants 132380 and 252018), and by the Research Funds of the University of Helsinki. Part of this work was done while JR and JS were affiliated with the University of Helsinki.

CL: This material is based upon work supported by the National Science Foundation under Grant Nos.\ CCF-AF-0937274, CNS-1035199, 0939370-CCF and CCF-1217506, the AFOSR under Award number FA9550-13-1-0042, and the German Research Foundation (DFG, reference number Le 3107/1-1).

Computer resources were provided by the Aalto University School of Science ``Science-IT'' project, and by the Department of Computer Science at the University of Helsinki.

\bibliographystyle{plain}
\bibliography{counting}

\newpage
\appendix
\section{Algorithm Listings}\label{app:alg}

In this appendix, we give two examples of our algorithms---machine-readable versions of all algorithms, verification code, and some illustrations are available online~\cite{results}.

\tableref{tab:alg-3-4-1-7-c} gives a cyclic algorithm for $n=4$. The rows are labelled with $(x_0,x_1)$, the columns are labelled with $(x_2,x_3)$, and the values indicate $A_0((x_0,x_1,x_2,x_3))$, that is, the new state of the first node in the observed configuration $\vec x$. The projection graph (\sectionref{sec:projgraph}) for this algorithm is given in \figureref{fig:alg-3-4-1-7-c}.

\tableref{tab:alg-2-6-1-6} shows a non-cyclic algorithm for $n=6$. Again, the rows are labelled with the first half $(x_0,x_1,x_2)$ of the observed state $\vec x$ and the columns are labelled with the second half $(x_3,x_4,x_5)$ of the observed state $\vec x$. The values show the new state for each node: $A_0(\vec x), A_1(\vec x), \dotsc, A_5(\vec x)$.

\begin{table}[h]
\centering
\begin{tabular}{@{}l@{\qquad}c@{\quad}c@{\quad}c@{\quad}c@{\quad}c@{\quad}c@{\quad}c@{\quad}c@{\quad}c@{}}
\toprule
 & 00 & 01 & 02 & 10 & 11 & 12 & 20 & 21 & 22 \\
\midrule
00 & 1 & 1 & 1 & 1 & 0 & 1 & 1 & 1 & 1 \\
01 & 1 & 1 & 1 & 2 & 2 & 0 & 1 & 1 & 1 \\
02 & 1 & 1 & 1 & 1 & 0 & 1 & 1 & 1 & 1 \\
10 & 1 & 0 & 1 & 1 & 0 & 0 & 1 & 0 & 1 \\
11 & 0 & 0 & 0 & 0 & 0 & 0 & 0 & 0 & 0 \\
12 & 1 & 0 & 1 & 0 & 0 & 0 & 0 & 0 & 0 \\
20 & 1 & 1 & 1 & 1 & 1 & 0 & 1 & 1 & 1 \\
21 & 1 & 1 & 1 & 1 & 0 & 0 & 1 & 0 & 0 \\
22 & 1 & 1 & 1 & 1 & 0 & 0 & 1 & 0 & 1 \\
\bottomrule
\end{tabular}
\caption{Cyclic algorithm for $s = 3$, $n = 4$, $f = 1$, and $t = 7$.}\label{tab:alg-3-4-1-7-c}
\end{table}

\begin{table}[h]
\centering
\begin{tabular}{@{}l@{\qquad}c@{\quad}c@{\quad}c@{\quad}c@{\quad}c@{\quad}c@{\quad}c@{\quad}c@{}}
\toprule
 & 000 & 001 & 010 & 011 & 100 & 101 & 110 & 111 \\
\midrule
000 & 111111 & 111111 & 111111 & 111111 & 111111 & 111111 & 111111 & 011000 \\
001 & 111111 & 111111 & 111111 & 111011 & 111011 & 111011 & 010001 & 010000 \\
010 & 111111 & 111111 & 111111 & 101001 & 111111 & 101001 & 011111 & 001000 \\
011 & 111111 & 111011 & 101001 & 100000 & 100001 & 100000 & 000001 & 000000 \\
100 & 111111 & 111111 & 111111 & 110110 & 111111 & 110110 & 011111 & 000000 \\
101 & 111111 & 111111 & 110110 & 110110 & 110110 & 110110 & 010000 & 000000 \\
110 & 011111 & 110110 & 011111 & 000000 & 011111 & 000000 & 011111 & 001000 \\
111 & 010110 & 010110 & 000000 & 000000 & 000010 & 000000 & 000001 & 000000 \\
\bottomrule
\end{tabular}
\caption{Algorithm for $s = 2$, $n = 6$, $f = 1$, and $t = 6$.}\label{tab:alg-2-6-1-6}
\end{table}

\begin{figure}
    \centering
    \includegraphics[page=1]{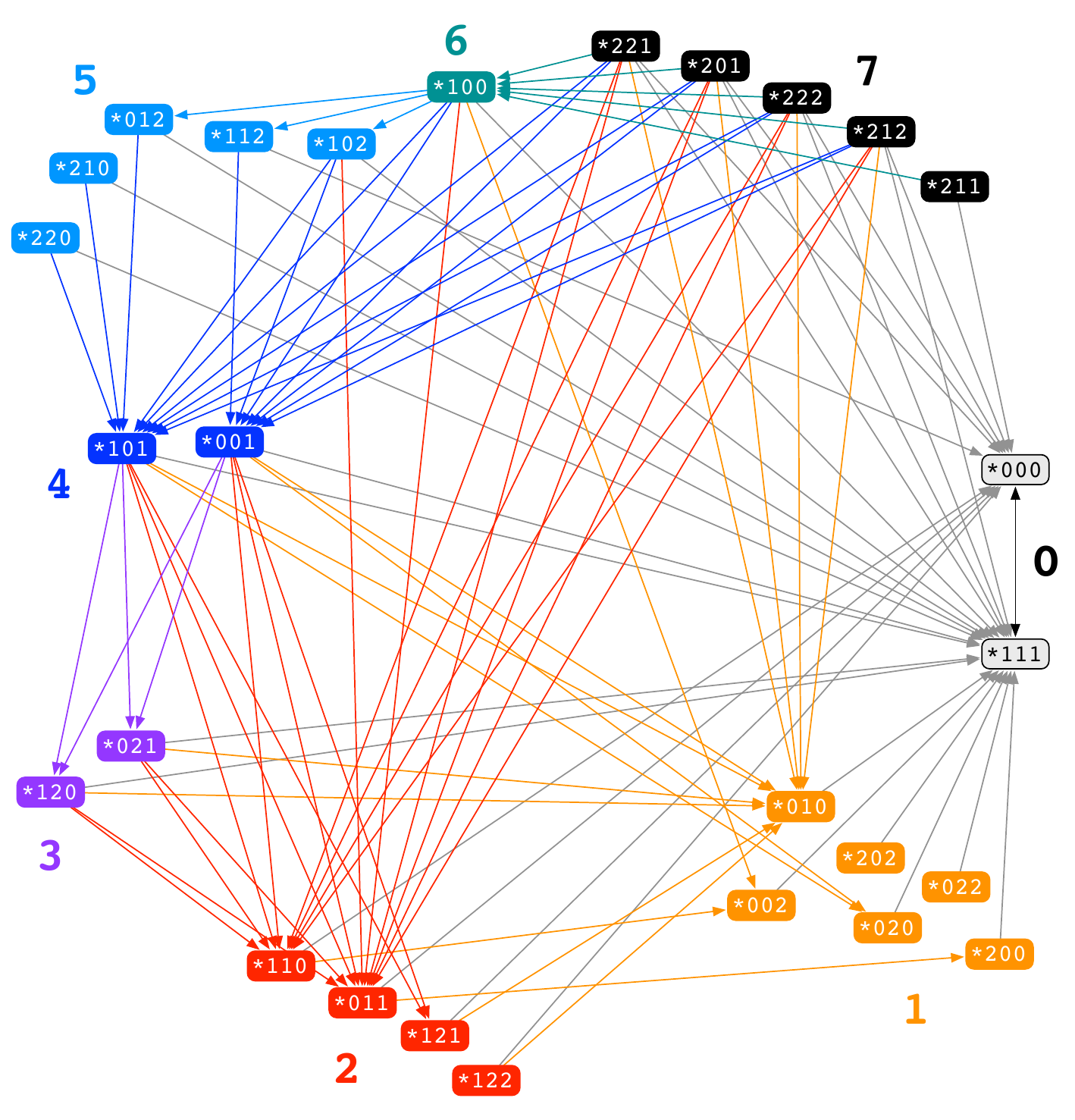}
    \caption{The projection graph $G_F(\vec A)$ for the algorithm $\vec A$ given in \tableref{tab:alg-3-4-1-7-c}, assuming that the faulty nodes are $F = \{0\}$. The actual configurations have been clustered according to the length of the longest path that avoids the good states $\vec 0_F$ and $\vec 1_F$. Based on the projection graph, it is straightforward to verify that for any initial state and for any adversarial activities the algorithm will stabilise in $t = 7$ steps.}\label{fig:alg-3-4-1-7-c}
\end{figure}

\end{document}